\documentclass[11pt]{article}
\usepackage{algorithm, algpseudocode}
\usepackage[utf8]{inputenc}

\usepackage{graphicx}
\usepackage{subfigure}
\usepackage{multirow}
\usepackage{comment}
\usepackage{amsmath,amssymb}
\usepackage{amsthm}
\usepackage[T1]{fontenc}
\usepackage{authblk}

\usepackage{MnSymbol}

\newtheorem{theorem}{Theorem}
\newtheorem{definition}[theorem]{Definition}
\newtheorem{lemma}[theorem]{Lemma}

\theoremstyle{remark}

\newcommand*{\pron}[1]{{{\cal P}_{n=#1}}}

\newcommand*{\proP}[1]{{{\cal P}_{P=#1}}}

\newcommand*{\Lleft}{{L_{l}}}
\newcommand*{\Lright}{{L_{r}}}
\newcommand*{\Lleader}{{L_{se}}} 
\newcommand*{\Lsleft}{{L^t_{l}}}
\newcommand*{\Lsright}{{L^t_{r}}}
\newcommand*{\Lsleader}{{L_{se'}}}
\newcommand*{\Leleader}{{L^t_{se}}}
\newcommand*{\Lesleader}{{L^t_{se'}}}

\newcommand*{\C}[2]{{C^{#1}_{#2}}}
\newcommand*{\Ch}[2]{{\hat{C}^{#1}_{#2}}}
\newcommand*{\tre}{{tre}}
\newcommand*{\F}{{\phi}}

\title{Population Protocols for Graph Class Identification Problems} 

\author[1]{Hiroto Yasumi}
\author[1]{Fukuhito Ooshita}
\author[1]{Michiko Inoue}

\affil[1]{Nara Institute of Science and Technology}
\date{}

\begin{document}

\maketitle

\begin{abstract}
In this paper, we focus on graph class identification problems in the population protocol model.
A graph class identification problem aims to decide whether a given communication graph is in the desired class (e.g. whether the given communication graph is a ring graph). Angluin et al. proposed graph class identification protocols with directed graphs and designated initial states under global fairness [Angluin et al., DCOSS2005]. We consider graph class identification problems for undirected graphs on various assumptions such as initial states of agents, fairness of the execution, and initial knowledge of agents. In particular, we focus on lines, rings, $k$-regular graphs, stars, trees, and bipartite graphs. 
With designated initial states, we propose graph class identification protocols for $k$-regular graphs, and trees under global fairness, and propose a graph class identification protocol for stars under weak fairness. Moreover, we show that, even if agents know the number of agents $n$, there is no graph class identification protocol for lines, rings, $k$-regular graphs, trees, or bipartite graphs under weak fairness. 
On the other hand, with arbitrary initial states, we show that there is no graph class identification protocol for lines, rings, $k$-regular graphs, stars, trees, or bipartite graphs. 
\end{abstract}

\section{Introduction}

\paragraph*{Background and Motivation}
The population protocol model is an abstract model for low-performance devices, introduced by Angluin et al.~\cite{angluin2006computation}. In this model, a network, called population, consists of multiple devices called agents. Those agents are anonymous (i.e., they do not have identifiers), and move unpredictably (i.e., they cannot control their movements). When two agents approach, they are able to communicate and update their states (this communication is called an interaction). By a sequence of interactions, the system proceeds a computation. 
In this model, there are various applications such as sensor networks used to monitor wild birds and molecular robot networks~\cite{murata2013molecular}. 

In this paper, we study the computability of graph properties of communication graphs in the population protocol model. Concretely, we focus on graph class identification problems that aim to decide whether the communication graph is in the desired graph class. In most distributed systems, it is essential to understand properties of the communication graph in order to design efficient algorithms. Actually, in the population protocol model, efficient protocols are proposed with limited communication graphs (e.g., ring graphs and regular graphs)~\cite{alistarh2021fast,angluin2005self,chen2019self,chen2020self}. 
 In the population protocol model, the computability of the graph property was first considered in \cite{angluin2005stably}. In \cite{angluin2005stably}, Angluin et al. proposed various graph class identification protocols with directed graphs and designated initial states under global fairness. Concretely, Angluin et al. proposed graph class identification protocols for directed lines, directed rings, directed stars, and directed trees. Moreover, they proposed graph class identification protocols for other graphs such as 1) graphs having degree bounded by a constant $k$, 2) graphs containing a fixed subgraph, 3) graphs containing a directed cycle, and 4) graphs containing a directed cycle of odd length. 
 However, there are still some open questions such as ``What is the computability for undirected graphs?'' and ``How do other assumptions (e.g., initial states, fairness, etc.) affect the computability?'' In this paper, we answer those questions. That is, we clarify the computability of graph class identification problems for undirected graphs under various assumptions such as initial states of agents, fairness of the execution, and an initial knowledge of agents. 

We remark that some protocols in \cite{angluin2005stably} for directed graphs can be easily extended to undirected graphs with designated initial states under global fairness (see Table \ref{tab:graph}).
Concretely, graph class identification protocols for directed lines, directed rings, and directed stars can be easily extended to protocols for undirected lines, undirected rings, and undirected stars, respectively.
In addition, the graph class identification protocol for bipartite graphs can be deduced from the protocol that decides whether a given graph contains a directed cycle of odd length. 
This is because, if we replace each edge of an undirected non-bipartite graph with two opposite directed edges, the directed non-bipartite graph always contains a directed cycle of odd length.
On the other hand, the graph class identification protocol for directed trees cannot work for undirected trees because the protocol uses a property of directed trees such that in-degree (resp., out-degree) of each agent is at most one on an out-directed tree (resp., an in-directed tree). 
Note that agents can identify trees if they understand the graph contains no cycle.
However, the graph class identification protocol for graphs containing a directed cycle in directed graphs cannot be used to identify a (simple) cycle in undirected graphs.
This is because, if we replace an undirected edge with two opposite directed edges, the two directed edges compose a directed cycle.

\paragraph*{Our Contributions}
In this paper, we clarify the computability of graph class identification problems for undirected graphs under various assumptions.
A summary of our results is given in Table \ref{tab:graph}.
We propose a graph class identification protocol for trees with designated initial states under global fairness. This protocol works with constant number of states even if no initial knowledge is given. Moreover, under global fairness, we also propose a graph class identification protocol for $k$-regular graphs with designated initial states. 
On the other hand, under weak fairness, we show that there exists no graph class identification protocol for lines, rings, $k$-regular graphs, stars, trees, or bipartite graphs even if the upper bound $P$ of the number of agents is given. Moreover, in the case where the number of agents $n$ is given, we propose a graph class identification protocol for stars and prove that there exists no graph class identification protocol for lines, rings, $k$-regular graphs, trees, or bipartite graphs. 
With arbitrary initial states, we prove that there is no protocol for lines, rings, $k$-regular graphs, stars, trees, or bipartite graphs. 

In this paper, because of space constraints, we omit the details of protocols (see the full version in the appendix).

\paragraph*{Related Works}
In the population protocol model, researchers studied various fundamental problems such as leader election problems~\cite{alistarh2015polylogarithmic,berenbrink2020optimal,doty2018stable,gkasieniec2019almost}, counting problems~\cite{aspnes2017time,beauquier2015space,beauquier2007self}, majority problems~\cite{angluin2008simple,berenbrink2020time,gasieniec2017deterministic}, etc. 
In \cite{alistarh2021fast,angluin2005self,chen2019self,chen2020self}, researchers proposed efficient protocols for such fundamental problems with limited communication graphs. 
More concretely, Angluin et al. proposed a protocol that constructs a spanning tree with regular graphs~\cite{angluin2005self}. Chen et al. proposed self-stabilizing leader election protocols with ring graphs~\cite{chen2019self} and regular graphs~\cite{chen2020self}. Alistarh et~al. showed that protocols for complete graphs (including the leader election protocol, the majority protocol, etc.) can be simulated efficiently in regular graphs~\cite{alistarh2021fast}. 

For the graph class identification problem, Chatzigiannakis et al. studied solvabilities for directed graphs with some properties on the mediated population protocol model~\cite{chatzigiannakis2010stably}, where the mediated population protocol model is an extension of the population protocol model. In this model, a communication link (on which agents interact) has a state. Agents can read and update the state of the communication link when agents interact on the communication link. 
In \cite{chatzigiannakis2010stably}, they proposed graph class identification protocol for some graphs such as 1) graphs having degree bounded by a constant $k$, 2) graphs in which the degree of each agent is at least $k$, 3) graphs containing an agent such that in-degree of the agent is greater than out-degree of the agent, 4) graphs containing a directed path of at least $k$ edges, etc. 
Since Chatzigiannakis et~al. proposed protocols for the mediated population protocol model, the protocols cannot work in the population protocol model. 
As impossibility results, they showed that there is no graph class identification protocol that decides whether the given directed graph has two edges $(u,v)$ and $(v,u)$ for two agents $u$ and $v$, or whether the given directed graph is weakly connected. 

As another perspective of communication graphs, Michail and Spirakis proposed a network constructors model that is an extension of the mediated population protocol~\cite{michail2016simple}. 
The network constructors model aims to construct a desired graph on the complete communication graph by using communication links with two states. Each communication link only has \emph{active} or \emph{inactive} state. Initially, all communication links have inactive state. By activating/deactivating communication links, the protocol of this model constructs a desired communication graph that consists of agents and activated communication links. In \cite{michail2016simple}, they proposed protocols that construct spanning lines, spanning rings, spanning stars, and regular graphs. Moreover, by relaxing the number of states, they proposed a protocol that constructs a large class of graphs. 

\begin{table}[t!]
\caption{The number of states to solve the graph class identification problems. $n$ is the number of agents and $P$ is an upper bound of the number of agents}
\begin{center}
\label{tab:graph}
\scalebox{0.85}{
\begin{tabular}{|c|c|c|c|c|c|c|c|c|c|}
\hline
\multicolumn{3}{|c|}{Model} & \multicolumn{6}{c|}{Graph Properties} \\
\hline
Initial states & Fairness & Initial knowledge & \emph{Line} & \emph{Ring} & \emph{Bipartite} & \emph{Tree} & \emph{$k$-regular} & \emph{Star} \\
\hline
\multirow{5}{*}{Designated} & \multirow{3}{*}{Global} & $n$ & $O(1)$\textdagger & $O(1)$\textdagger & $O(1)$\textdagger & $O(1)${\bf *} & $O(k \log n)${\bf *} & $O(1)$\textdagger \\
\cline{3-9}
 &  & $P$ & $O(1)$\textdagger & $O(1)$\textdagger &  $O(1)$\textdagger & $O(1)${\bf *} & $O(k \log P)${\bf *} & $O(1)$\textdagger \\
\cline{3-9} 
 &  & None & $O(1)$\textdagger & $O(1)$\textdagger & $O(1)$\textdagger & $O(1)${\bf *} & - & $O(1)$\textdagger \\
\cline{2-9} 
 & \multirow{2}{*}{Weak} & $n$ & \multicolumn{5}{c|}{Unsolvable{\bf *}} & $O(n)${\bf *} \\ 
\cline{3-9}
 &  & $P$/None & \multicolumn{6}{c|}{Unsolvable{\bf *}} \\
\hline
Arbitrary & Global/Weak & $n$/$P$/None & \multicolumn{6}{c|}{Unsolvable{\bf *}} \\
\hline
\multicolumn{9}{r}{\textbf{*} Contributions of this paper~~~\textdagger Deduced from Angluin et al.~\cite{angluin2005stably}}\\ 
\end{tabular}
}
\end{center}
\end{table}

\section{Definitions}
 \subsection{Population Protocol Model} 
A communication graph of a population is represented by a simple undirected connected graph $G=(V,E)$, where $V$ represents a set of agents, and $E \subseteq V \times V$ is a set of edges (containing neither multi-edges nor self-loops) that represent the possibility of an interaction between two agents (i.e., only if $(a,b) \in E$ holds, two agents $a \in V$ and $b \in V$ can interact).

A protocol ${\cal P}=(Q, Y,  \gamma, \delta)$ consists of a set $Q$ of possible states of agents, a finite set of output symbols $Y$, an output function $\gamma: Q \rightarrow Y$, and a set of transitions $\delta$ from $Q\times Q$ to $Q\times Q$. 
 Output symbols in $Y$ represent outputs as the results according to the purpose of the protocol. Output function $\gamma$ maps a state of an agent to an output symbol in $Y$. 
Each transition in $\delta$ is denoted by $(p, q) \rightarrow (p', q')$. This means that, when an agent $a$ in state $ p $ interacts with an agent $b$ in state $ q $, their states become $ p'$ and $ q' $, respectively. We say that such $a$ is an initiator and such $b$ is a responder. When $a$ and $b$ interact as an initiator and a responder, respectively, we simply say that $a$ interacts with $b$.
Transition $(p,q) \rightarrow (p',q')$ is null if both $p=p'$ and $q=q'$ hold. We omit null transitions in the  descriptions of protocols.
Protocol ${\cal P}=(Q, Y,  \gamma, \delta)$ is deterministic if, for any pair of states $(p,q)\in Q\times Q$, exactly one transition $(p,q)\rightarrow (p', q') $ exists in $\delta$. 
We consider only deterministic protocols in this paper. 

A configuration represents a global state of a population, defined as a vector of states of all agents.
 A state of agent $a$ in configuration $C$ is denoted by $s(a, C)$. 
Moreover, when $C$ is clear from the context, we simply use $s(a)$ to denote the state of agent $a$. 
A transition from configuration $C$ to configuration $C'$ is denoted by $C \rightarrow C'$, and means that, by a single interaction between two agents, configuration $C'$ is obtained from configuration $C$. 
 For two configurations $C$ and $C'$, if there exists a sequence of configurations $C = C_0$, $C_1$, $\ldots$, $C_m = C'$ such that $C_i \rightarrow C_{i+1}$ holds for every $i$ ($0 \le i < m$), we say 
 $C'$ is reachable from $C$, denoted by $C \xrightarrow{*} C'$.
 
 An execution of a protocol is an infinite sequence of configurations $\Xi =C_0, C_1, C_2, \ldots$ where $C_i \rightarrow C_{i+1}$ holds for every $i$ ($i \ge 0$).
An execution $\Xi$ is weakly-fair if, for any adjacent agents $a$ and $a'$, $a$ interacts with $a'$ and $a'$ interacts with $a$ infinitely often\footnote{We use this definition only for the lower bound under weak fairness. For the upper bound, we use a slightly weaker assumption. We show that our proposed protocol for weak fairness works if, for any adjacent agents $a$ and $a'$, $a$ and $a'$ interact infinitely often (i.e., it is possible that, for any interaction between some adjacent agents $a$ and $a'$, $a$ becomes an initiator and $a'$ never becomes an initiator). Note that, in the protocol, if a transition $(p,q) \rightarrow (p',q')$ exists for $p \neq q$, a transition $(q,p) \rightarrow (q',p')$ also exists. }. 
An execution $\Xi$ is globally-fair if, for each pair of configurations $C$ and $C'$ such that $C \rightarrow C'$, $C'$ occurs infinitely often when $C$ occurs infinitely often. 
Intuitively, global fairness guarantees that, if configuration $C$ occurs infinitely often, then
any possible interaction in $C$ also occurs infinitely often. Then, if $C$ occurs infinitely often, $C'$ satisfying $C\rightarrow C'$ occurs infinitely often, and we can deduce that $C''$ satisfying $C'\rightarrow C''$ also occurs infinitely often. Overall, with global fairness, if a configuration $C$ occurs infinitely often, then every configuration $C^*$ reachable from $C$ also occurs infinitely often.


In this paper, we consider three possibilities for an initial knowledge of agents: the number of agents $n$, the upper bound $P$ of the number of agents, and no knowledge. 
Note that the protocol depends on this initial knowledge. 
When we explicitly state that an integer $x$ is given as the number of agents, we write the protocol as $\pron{x}$. 
Similarly, when we explicitly state that an integer $x$ is given as the upper bound of the number of agents, the protocol is denoted by $\proP{x}$. 

\subsection{Graph Properties and Graph Class Identification Problems}
We define graph properties treated in this paper as follows:
\begin{itemize}
\item A graph $G$ satisfies property \emph{tree} if there is no cycle on graph $G$. 
\item A graph $G=(V,E)$ satisfies property \emph{$k$-regular} if the degree of every agent in $V$ is $k$. 
\item A graph $G$ satisfies property \emph{star} if $G$ is a tree with one internal agent and $n-1$ leaves. 
\item A graph $G=(V,E)$ satisfies property \emph{bipartite} if $V$ can be divided into two disjoint and independent sets $U$ and $W$. 
\item A graph $G=(V, E)$ satisfies property \emph{line} if $E=\{(v_0,v_1)$, $(v_1,v_2)$, $(v_2,v_3)$, $\ldots$, $(v_{n-1}, v_n)\}$ for $V = \{v_1$, $v_2$, $\ldots$ $v_n\}$. 
\item A graph $G=(V,E)$ satisfies property \emph{ring} if the degree of every agent in $V$ is $2$. 
\end{itemize} 
Let $gp$ be an arbitrary graph property. The $gp$ identification problem aims to decide whether a given communication graph $G$ satisfies property $gp$. 
In the $gp$ identification problem, the output set is $Y = \{yes, no\}$. Recall that the output function $\gamma$ maps a state of an agent to an output symbol in $Y$ (i.e., $yes$ or $no$). 
A configuration $C$ is stable if $C$ satisfies the following conditions: There exists $yn \in \{yes,no\}$ such that 1) $\forall a \in V:\gamma(s(a,C))=yn$ holds, and
2) for every configuration $C'$ such that $C \xrightarrow{*} C'$, $\forall a \in V:\gamma(s(a,C'))=yn$ holds. 

  An execution $\Xi = C_0$, $C_1$, $C_2$, $\ldots$ solves the $gp$ identification problem if $\Xi$ includes a stable configuration $C_t$ that satisfies the following conditions.
\begin{enumerate}
	\item If a given graph $G=(V,E)$ satisfies graph property $gp$, $\forall a \in V:\gamma(s(a,C_{t}))=yes$ holds. 
	\item If a given graph $G=(V,E)$ does not satisfy graph property $gp$, $\forall a \in V:\gamma(s(a,C_{t}))=no$ holds. 
\end{enumerate}

  A protocol $ {\cal P} $ solves the $gp$ identification problem if every possible execution $ \Xi $ of protocol $ {\cal P} $ solves the $gp$ identification problem. 


\section{Graph Class Identification Protocols}

\subsection{Tree Identification Protocol with No Initial Knowledge under Global Fairness} 
In this section, we give a tree identification protocol (hereinafter referred to as ``TI protocol'') with 18 states and designated initial states under global fairness.

The basic strategy of the protocol is as follows.
First, agents elect one leader token, one right token, and one left token. Agents carry these tokens on a graph by interactions as if each token moves freely on the graph.
After the election, agents repeatedly execute a trial to detect a cycle by using the tokens. 
The trial starts when two adjacent agents $x$ and $y$ have the right token and the left token, respectively. 
During the trial, $x$ and $y$ hold the right token and the left token, respectively. 
To detect a cycle, agents use the right token and the left token as a single landmark. 
The right token and the left token correspond to a right side and a left side of the landmark, respectively.
If agents can carry the leader token from the right side of the landmark to the left side of the landmark without passing through the landmark, the trial succeeds.
Clearly, when the trial succeeds, there is a cycle.
In this case, an agent with the leader token recognizes the success of the trial and decides that there is a cycle and thus the given graph is not a tree.
Then, the decision is conveyed to all agents by the leader token and thus all agents decide that the given graph is not a tree.
Initially, all agents think that the given graph is a tree. Hence, unless the trial succeeds, all agents continue to think that the given graph is a tree. 
Therefore, the protocol solves the problem.

Before we explain the details of the protocol, first we introduce variables at agent $a$. 
\begin{itemize}
\item $LF_a \in \{\Lleader$, $\Lleft$, $\Lright$, $\Leleader$, $\Lesleader$, $\Lsleader$, $\Lsleft$, $\Lsright$, $\F \}$: Variable $LF_a$, initialized to $\Lright$, represents a token held by agent $a$. 
If $LF_a$ is not $\F$, agent $a$ has $LF_a$ token.
There are three types of tokens: a leader token ($\Lleader$, $\Leleader$, $\Lesleader$, and $\Lsleader$), a left token ($\Lleft$ and $\Lsleft$), and a right token ($\Lright$ and $\Lsright$).
We show the details of them later. 
$\F$ represents no token.
\item $\tre_a \in \{yes$, $no \}$: Variable $\tre_a$, initialized to $yes$, represents a decision of the tree. 
If $\tre_a = yes$ holds for agent $a$, then $\gamma(s_a) = yes$ holds (i.e., $a$ decides that the given graph is a tree).
If $\tre_a = no$ holds, then $\gamma(s_a) = no$ holds (i.e., $a$ decides that the given graph is not a tree).
\end{itemize}
The protocol uses 18 states because the number of values taken by variable $LF_a$ is 9 and the number of values taken by variable $\tre_a$ is 2.

\begin{algorithm}[t!]
	\caption{A TI protocol (1/2)}         
	\label{alg:cyc1}
	\algblockdefx{when}{End}{\textbf{when}}{\textbf{end}}
	\begin{algorithmic}[1]
		\algrenewcommand{\algorithmicrequire}{\textbf{Variables at an agent $a$:}}
		\Require
		\Statex $LF_a \in \{\Lleader$, $\Lleft$, $\Lright$, $\Leleader$, $\Lesleader$, $\Lsleader$, $\Lsleft$, $\Lsright$, $\F \}$: Token held by the agent, initialized to $\Lright$.
       \Statex $\tre_a \in \{yes$, $no \}$: Decision of the tree, initialized to $yes$. 
		\renewcommand{\algorithmicwhile}{\textbf{when}} 
		\While {agent $a$ interacts with agent $b$}
		\Statex \{ The election of tokens \}
		\If{$LF_a, LF_b \in \{ \Lsright$, $\Lright\}$}
		\State $LF_b \leftarrow \Lleft$
		\ElsIf{$LF_a, LF_b \in \{ \Lsleft$, $\Lleft\}$}
		\State $LF_b \leftarrow \Lleader$
		\ElsIf{$LF_a, LF_b \in \{ \Lleader$, $\Leleader$, $\Lesleader$, $\Lsleader\}$}
		\State $LF_a \leftarrow \Lleader$, $LF_b \leftarrow \F$
		\State $\tre_a \leftarrow yes$ 
		\Statex \{ Movement of tokens \}
		\ElsIf{$LF_a \neq \F$ $\wedge$ $LF_b = \F$}
		\If{$LF_a \in \{ \Lleader$, $\Leleader$, $\Lesleader$, $\Lsleader\}$}
		\State $\tre_b \leftarrow \tre_a$
		\EndIf
		\If{$LF_a = L^t_{\kappa}$ for $\kappa \in \{l, r\}$}
		\State $LF_a \leftarrow L_\kappa$
		\ElsIf{$LF_a = \Lsleader$ $\vee$ $LF_a = \Lesleader$}
		\State $LF_a \leftarrow \Lleader$
		\EndIf
		\State $LF_a \leftrightarrow LF_b$ {\bf *}
		\Statex \{ Decision \}
		\ElsIf{$LF_a = \Lleader$ $\wedge$ $LF_b = \Lleft$}
		\State $LF_a \leftarrow \Lsleft$, $LF_b \leftarrow \Lsleader$
		\State $\tre_b \leftarrow \tre_a$
		\ElsIf{$LF_a = \Lsleader$ $\wedge$ $LF_b = \Lright$}
		\State $LF_a \leftarrow \Lsright$, $LF_b \leftarrow \Leleader$
		\State $\tre_b \leftarrow \tre_a$
		\ElsIf{$LF_a = \Leleader$ $\wedge$ $LF_b = \Lsleft$}
		\State $LF_a \leftarrow \Lleft$, $LF_b \leftarrow \Lesleader$ 
		\State $\tre_b \leftarrow \tre_a$
		\ElsIf{$LF_a = \Lesleader$ $\wedge$ $LF_b = \Lsright$}
		\State $LF_a \leftarrow \Lright$, $LF_b \leftarrow \Lleader$
		\State $\tre_b \leftarrow no$
		\Statex
		\Statex {\bf *} $p \leftrightarrow q$ means that $p$ and $q$ exchange values.
		\Statex \Comment Continued on the next page
		\algstore{break}
	\end{algorithmic}
\end{algorithm}

\setcounter{algorithm}{0} 

\begin{algorithm}[t!]
	\caption{A TI protocol (2/2)}         
	\label{alg:cyc2}
	\algblock{when}{End}
	\begin{algorithmic}[1]
		\renewcommand{\algorithmicwhile}{\textbf{when}}
		\algrestore{break}
		\ElsIf{$LF_a \neq \F$ $\wedge$ $LF_b \neq \F$}
		\If{$LF_{ab} \in \{ \Lleader$, $\Leleader$, $\Lesleader$, $\Lsleader\}$ for $ab \in \{a,b\}$}
		\State $\tre_a \leftarrow \tre_{ab}$, $\tre_b \leftarrow \tre_{ab}$ 
		\EndIf
		\If{$LF_{ab} = L^t_{\kappa}$ for $ab \in \{a,b\}$ and $\kappa \in \{l, r\}$}
		\State $LF_{ab} \leftarrow L_\kappa$
		\EndIf
		\If{$LF_{ab} = \Lsleader$ $\vee$ $LF_{ab} = \Leleader$ $\vee$ $LF_{ab} = \Lesleader$ for $ab \in \{a,b\}$}
		\State $LF_{ab} \leftarrow \Lleader$
		\EndIf
		\State $LF_a \leftrightarrow LF_b$
		\EndIf
		\renewcommand{\algorithmicwhile}{}
		\EndWhile		
	\end{algorithmic}
\end{algorithm}

\setcounter{algorithm}{1} 

From now, we explain the details of the protocol. The protocol is given in Algorithms~\ref{alg:cyc1} and \ref{alg:cyc2}. 
\paragraph*{Election of three tokens (lines 2--8)}
Initially, each agent has a right token. When two agents with right tokens interact, the agents change one of the tokens to a left token (lines 2--3).
When two agents with left tokens interact, the agents change one of the tokens to a leader token (lines 4--5).
When two agents with leader tokens interact, the agents delete one of the tokens (lines 6--7).
As we explain later, agents carry a token on a graph by interactions as if a token moves freely on the graph.
Thus, by the above behaviors, eventually agents elect one right token, one left token, and one leader token.

In the cycle detection part, we will just show behaviors after agents complete the token election (i.e., agents elect one right token, one left token, and one leader token). 
However, in this protocol, agents may make a wrong decision before agents complete the token election. 
Agents overcome this problem by the following behaviors.
\begin{itemize} 
\item Agents behave as if the leader token has the decision, and agents follow the decision. Concretely, when agent $a$ moves the leader token to agent $b$ by an interaction, agent $b$ copies $\tre_a$ to $\tre_b$.
Since the leader token moves freely on the graph, finally all agents follow the decision of the leader token.
\item 
When two agents with leader tokens interact and agents delete one of them, the agents reset $\tre$ of the remaining leader token. That is, if agent $a$ has the remaining leader token, it assigns $yes$ to $\tre_a$ (line 8).
\end{itemize}
Note that the last token is elected by an interaction between agents with the leader tokens (i.e., the last interaction in this election part occurs between agents with the leader tokens). 
By this interaction, the elected leader token resets its $\tre$ to $yes$. 
Hence, $\tre$ of the leader token is $yes$ just after agents complete the token election, and all agents follow $\tre$ of the leader token. 
Thus, because agents correctly detect a cycle after the token election (we will show this later), agents are not affected by the wrong decision. 

\paragraph*{Movement of tokens (lines 9--18)} 
When an agent having a token interacts with an agent having no token, the agents move the token (lines 9--18).
Concretely, the token moves by a behavior of line 18. In lines 10--12, $\tre$ of the leader token is conveyed.
We will explain the behavior of lines 13--17 after the explanation of the trial of the cycle detection.

\paragraph*{The trial of the cycle detection (lines 19--43)} 
In this paragraph, we show that, by a trial of the cycle detection, agents correctly detect a cycle after agents complete the token election. 
To begin with, we explain the start of the trial. 
To start the trial, agents place the left token and the right token next to each other. 
To distinguish between a moving token and a placed token, we use a trial mode. 
Agents regard right and left tokens in a trial mode as placed tokens. 
Thus, when agents place the right token and the left token, agents make the right token and the left token transition to the trial mode. 
An $\Lsright$ token (resp., an $\Lsleft$ token) represents the right token (resp., the left token) in the trial mode. 
An $\Lright$ token (resp., an $\Lleft$ token) represents the right token (resp., the left token) in a non-trial mode. 

An image of the start of the trial is shown in Figure \ref{fig:cycimg}. 
Figures \ref{fig:cycimg}(a) and \ref{fig:cycimg}(b) show the behavior such that agents make the left token and the right token transition to the trial mode. 
First, an agent with an $\Lleader$ token changes an $\Lleft$ token to an $\Lsleft$ token by an interaction (Figure \ref{fig:cycimg}(a)), where the $\Lleader$ token represents the default leader token. 
By the interaction, the agents exchange their tokens and the $\Lleader$ token transitions to an $\Lsleader$ token, where the $\Lsleader$ token represents the leader token next to the $\Lsleft$ token. 
This behavior appears in lines 19--21. 
Then, an agent with the $\Lsleader$ token changes the right token to a trial mode by an interaction (Figure \ref{fig:cycimg}(b)). 
By the interaction, the agents exchange their tokens. Thus, since the $\Lsleader$ token represents the leader token next to the $\Lsleft$ token, agents place an $\Lsright$ token next to the $\Lsleft$ token by the interaction. 
Hence, by the interaction, agents place the tokens in the following order: the $\Lsleft$ token, the $\Lsright$ token, the leader token (Figure \ref{fig:cycimg}(c)). 
Moreover, by the interaction, the $\Lsleader$ token transitions to an $\Leleader$ token, where the $\Leleader$ token represents the leader token trying to detect a cycle. 
This behavior appears in lines 22--24. 
When agents place all tokens as shown in Figure \ref{fig:cycimg}(c), a trial of the cycle detection starts. 

From now, we explain the main behavior of the cycle detection (Figure \ref{fig:cycsuc} and \ref{fig:cycfai}). 
Let $x$ (resp., $y$) be an agent having the $\Lsright$ token (resp., the $\Lsleft$ token). 
Let $\mathcal{X}$ (resp., $\mathcal{Y}$) be a set of agents adjacent to $x$ (resp., $y$). 
Let $\mathcal{X'} = \mathcal{X} \backslash \{y\}$ and $\mathcal{Y'} = \mathcal{Y} \backslash \{x\}$. 
In a trial, agents try to carry the leader token from an agent in $\mathcal{X'}$ to an agent in $\mathcal{Y'}$ without using the edge between $x$ and $y$. 

First, we explain the case where a trial succeeds (Figure \ref{fig:cycsuc}). 
In the trial, agents carry the $\Leleader$ token while the $\Lsright$ token and the $\Lsleft$ token are placed at $x$ and $y$, respectively. 
Concretely, if the following procedure occurs, the trial succeeds. 
\begin{enumerate}
\item Agents carry the $\Leleader$ token from an agent in $\mathcal{X'}$ to an agent in $\mathcal{Y'}$ without using the edge between $x$ and $y$ (Figure \ref{fig:cycsuc}(c)). 
\item An agent having the $\Leleader$ token interacts with agent $y$ having the $\Lsleft$ token. By the interaction, agents exchange their tokens and the $\Leleader$ token transitions to an $\Lesleader$ token (Figure \ref{fig:cycsuc}(d)). 
In addition, by the interaction, agents confirm that the $\Lsleft$ token was placed at $y$ while agents move the leader token to an agent in $\mathcal{Y'}$. 
The $\Lesleader$ token represents the leader token that confirmed it. 
This behavior appears in lines 25--27. 
\item Agent $y$ having the $\Lesleader$ token interacts with agent $x$ having the $\Lsright$ token (Figure \ref{fig:cycsuc}(e)). 
By the interaction, agents confirm that the $\Lsright$ token was placed at $x$ while agents move the leader token to an agent in $\mathcal{Y'}$. 
This behavior appears in lines 28--30. 
\end{enumerate} 
Clearly, if there is no cycle, agents do not perform this procedure. 
Thus, if agents perform this procedure, an agent with the leader token decides that there is a cycle and thus the given graph is not a tree (Figure \ref{fig:cycsuc}(f)). 
Concretely, the agent with the leader token changes its $\tre$ to $no$ (line 30). 

Next, we explain the case where a trial fails (Figure \ref{fig:cycfai}). 
There are three cases where the trial fails: (1) An agent having the $\Lsleader$ or $\Lesleader$ token fails to interact with the right token, (2) an agent having the $\Lsleft$ or $\Lsright$ token fails to wait for the leader token, and (3) an agent having the $\Leleader$ token fails to interact with an agent having the $\Lsleft$ token. 
Case (1) is that an agent having the $\Lsleader$ token (resp., the $\Lesleader$ token) interacts with an agent that does not have the $\Lright$ token (resp., the $\Lsright$ token). Figure \ref{fig:cycfai}(A-1) and (B-1) shows an example of case (1). 
By the interaction, agents make the token transition to the $\Lleader$ token (lines 9--17 and 31--43).
If agents make the $\Lsleader$ token transition to the $\Lleader$ token, the condition in line 22 is never satisfied in the trial. 
If agents make the $\Lesleader$ token transition to the $\Lleader$ token, the condition in line 28 is never satisfied in the trial. 
Case (2) is that an agent having an $\Lsleft$ token (resp., an $\Lsright$ token) interacts with an agent that does not have the $\Leleader$ token (resp., the $\Lesleader$ token). 
Figure \ref{fig:cycfai}(A-2) and (B-2) shows an example of case (2). 
By the interaction, agents make the $\Lsleft$ token (resp., the $\Lsright$ token) transition to an $\Lleft$ token (resp., an $\Lright$ token) by the behavior of lines 13--14 or 31--37, and thus the condition in line 25 or 28 is never satisfied in the trial.
Case (3) is that an agent having the $\Leleader$ token interacts with an agent having a token that is not the $\Lsleft$ token. Figure \ref{fig:cycfai}(A-3) and (B-3) shows an example of case (3). 
By the interaction, agents make the $\Leleader$ token transition to the $\Lleader$ token (lines 31--43). 
If agents make the $\Leleader$ token transition to the $\Lleader$ token, the condition in line 25 is never satisfied in the trial. 

Agents have an infinite number of chances of the trial.
This is because agents can make the leader token, the left token, and the right token transition to the $\Lleader$ token, the $\Lleft$ token, and the $\Lright$ token, respectively, from any configuration (lines 9--18 and 31--43).
Hence, from global fairness, eventually agents make the left and right tokens transition to the trial mode on the cycle and then agents find the cycle by the leader token.
Thus, eventually a trial succeeds if there is a cycle.

By the behaviors of the trial, since $\tre$ of the leader token is $yes$ just after agents complete the token election, $\tre$ of the leader token converges to a correct value. 
Since eventually all agents follow the decision of the leader token, all agents correctly decide whether the given graph is a tree or not. 

\begin{figure}[t!]
\begin{center}
\includegraphics[scale=0.35]{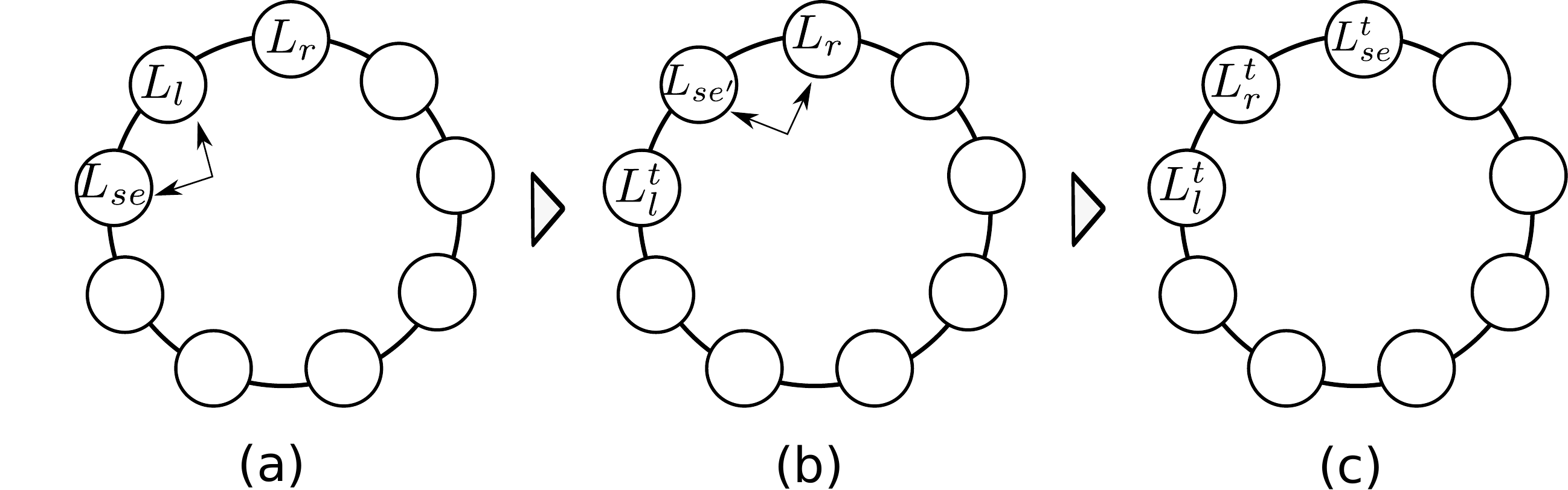}
\caption{An image of the start of the trial}
\label{fig:cycimg}
\end{center}
\end{figure}

\begin{figure}[t!]
\begin{center}
\includegraphics[scale=0.35]{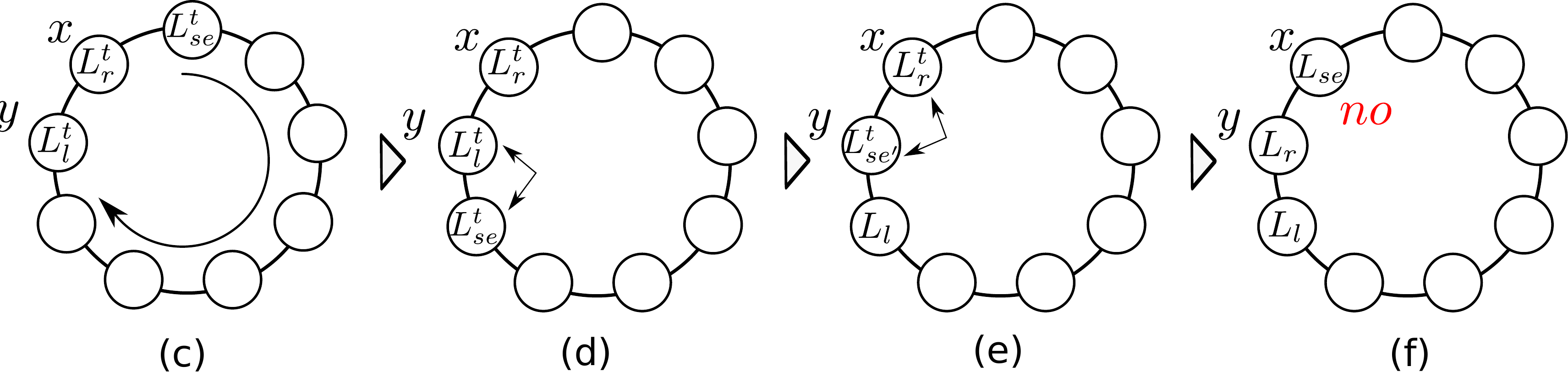}
\caption{An image of the success of the trial}
\label{fig:cycsuc}
\end{center}
\end{figure}

\begin{figure}[t!]
\begin{center}
\includegraphics[scale=0.35]{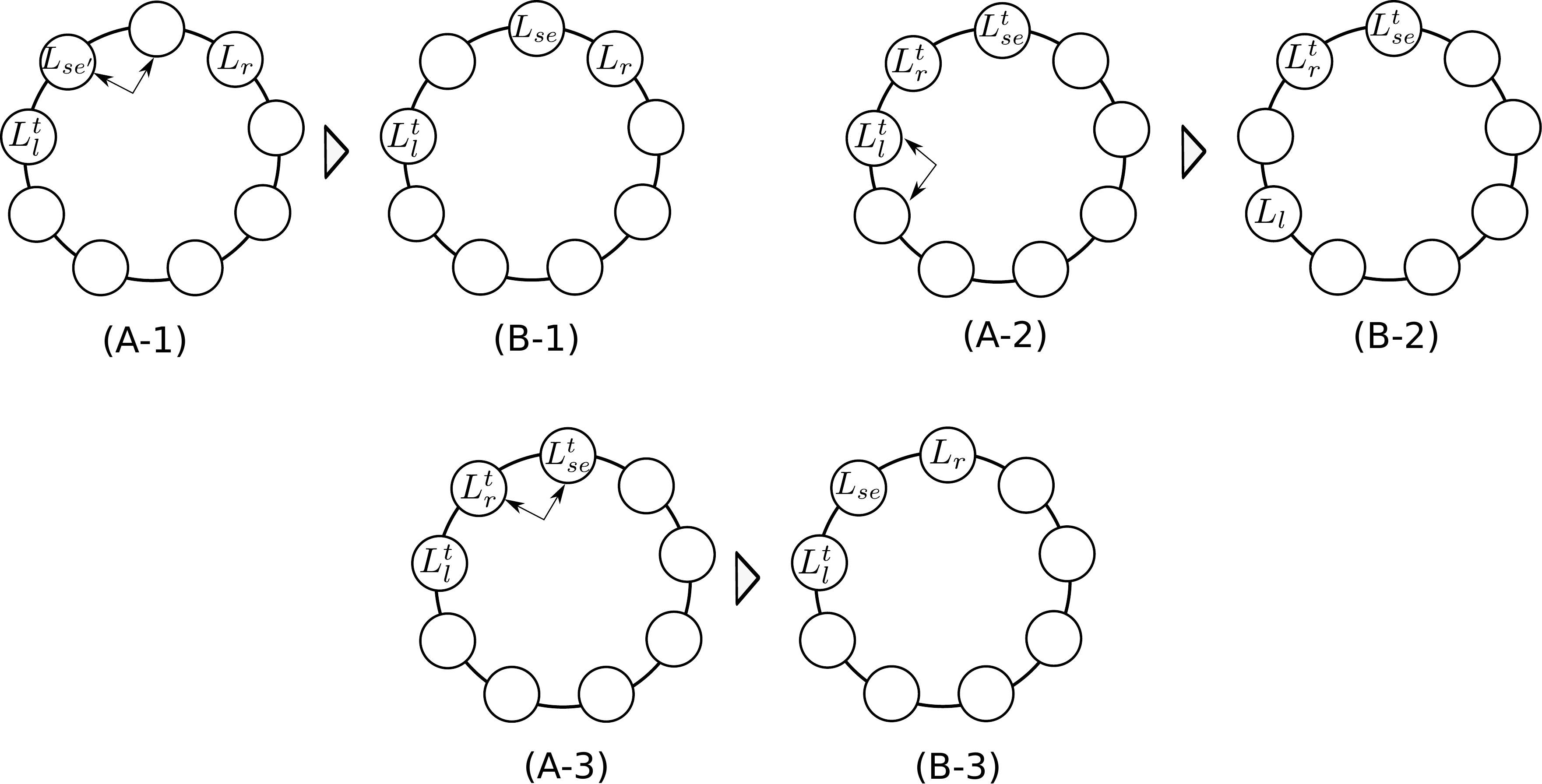}
\caption{Images of the fail of the trial}
\label{fig:cycfai}
\end{center}
\end{figure}

\subsection*{Correctness} 
First of all, if the number of agents $n$ is less than 3, clearly a leader token is not generated in Algorithm \ref{alg:cyc1}. Hence, in the case, $\tre_a$ of each agent $a$ converges to $yes$. Thus, since the given graph with $n < 3$ is a tree, each agent make a correct decision in this case. From now on, we consider the case where the number of agents $n$ is at least 3. 

To begin with, we define some notions for the numbers of the leader, left, and right tokens as follows:
\begin{definition}
The number of agents with $\Lright$ or $\Lsright$ tokens is denoted by $\# \Lright$.
The number of agents with $\Lleft$ or $\Lsleft$ tokens is denoted by $\# \Lleft$.
The number of agents with $\Lleader$, $\Leleader$, $\Lesleader$, or $\Lsleader$ tokens is denoted by $\# \Lleader$.
\end{definition}

Next, we define a configuration where agents complete the token election.
\begin{definition}
For an execution $\Xi = C_0$, $C_1$, $\ldots$, we say that agents complete the token election at $C_i$ if $\# \Lright > 1$, $\# \Lleft > 1$, or $\# \Lleader > 1$ holds in $C_{i-1}$, and $\# \Lright = 1$, $\# \Lleft=1$, and $\# \Lleader=1$ hold in $C_i$.
\end{definition}

From now, we show that agents eventually complete the token election, and, for agent $a$ with the leader token, $\tre_a = yes$ hold just after the election.
\begin{lemma}
\label{lem:onetoken}
For any globally-fair execution $\Xi = C_0$, $C_1$, $\ldots$, there is a configuration $C_i$ at which agents complete the token election. 

In $C_i$, there exists an agent $a$ that has an $\Lleader$ token and $\tre_a$ is $yes$. 
Moreover, in any configuration after $C_i$, $\# \Lright = 1$, $\# \Lleft = 1$, and $\# \Lleader = 1$ hold. 
\end{lemma}
\begin{proof}
Consider a globally-fair execution $\Xi = C_0$, $C_1$, $C_2$, $\ldots$.
From the pseudocode, when an agent having a leader token interacts with an agent having no leader token, agents move the leader token. 
Similarly, when an agent having a left token (resp., a right token) interacts with an agent having no left token (resp., no right token), agents move the token. 
Only if an agent having a leader token interacts with an agent having a leader token, the number of leader tokens decreases. 
Similarly, only if an agent having a left token (resp., a right token) interacts with an agent having a left token (resp., a right token), the number of the tokens decreases. 
These imply that, from global fairness, if there are two or more tokens of the same type (leader, left, or right), eventually adjacent agents have the tokens and then they interact. 

Hence, from global fairness, because there is no behavior to increase $\# \Lright$, $\# \Lright$ continues to decrease as long as $\# \Lright \ge 2$ holds, by the behavior of lines 2--3. 
Thus, after some configuration, $\# \Lright = 1$ holds and the behavior of lines 2--3 does not occur. 
After that, because there is no behavior to increase $\# \Lleft$ except for the behavior of lines 2--3, $\# \Lleft$ continues to decrease as long as $\# \Lleft \ge 2$ holds, by the behavior of lines 4--5.
Thus, after some configuration, $\# \Lleft = 1$ holds and the behavior of lines 4--5 does not occur. 
After that, because there is no behavior to increase $\# \Lleft$ except for the behavior of lines 4--5, $\# \Lleader$ continues to decrease as long as $\# \Lleader \ge 2$ holds, by the behavior of lines 6--7. 
Thus, after some configuration, $\# \Lleader = 1$ holds. 
Hence, there exists a configuration $C_i$ such that $\# \Lright = 1$, $\# \Lleft = 1$, and $\# \Lleader = 1$ hold after $C_i$ and $\# \Lright > 1$, $\# \Lleft > 1$, or $\# \Lleader > 1$ holds in $C_{i-1}$. 

If $n>3$ holds, agents execute lines 6--8 of the pseudocode at transition $C_{i-1} \rightarrow C_{i}$ because only the behavior of lines 6--8 decreases the number of leader tokens. 
For an agent $a$ with the leader token, the leader token transitions to an $\Lleader$ token and $\tre_a$ transitions to $yes$ when agents execute lines 6--8. 
If $n=3$ holds, agents execute lines 4--5 of the pseudocode at transition $C_{i-1} \rightarrow C_{i}$, and the first leader token is generated by this transition (and hence the leader token is the $\Lleader$ token and $\tre_a=yes$ for an agent $a$ with the $\Lleader$ token). These imply that, in $C_i$, there exists an agent $a$ that has the $\Lleader$ token and $\tre_a$ is $yes$. 
Therefore, the lemma holds.
\end{proof}

From Lemma \ref{lem:onetoken}, after agents complete the token election, only one leader token remains. 
From now on, we define $\tre$ of the leader token as $\tre_a$ such that agent $a$ has the leader token. 

From the pseudocode, $\tre$ of the leader token is conveyed to all agents. 
This implies that, after $\tre$ of the leader token converges, $\tre$ of each agent also converges to the same value as $\tre$ of the leader token.
Hence, from now, we show that $\tre$ of the leader token converges to $no$ (resp., $yes$) if there is a cycle (resp., no cycle) on the graph.

First, we show that $\tre$ of the leader token converges to $yes$ if there is no cycle on the graph. 
\begin{lemma}
\label{lem:nocyc}
For any globally-fair execution $\Xi$, if there is no cycle on a given communication graph, $\tre$ of the leader token converges to $yes$.
\end{lemma}
\begin{proof}
Variable $\tre$ of the leader token transitions to $no$ only if agents execute lines 28--30. 
From Lemma \ref{lem:onetoken}, when agents complete the token election, $\tre$ of the leader token transitions to $yes$. 
Thus, for the purpose of contradiction, we assume that, for a globally-fair execution $\Xi$ with a graph $G$ containing no cycle, agents execute lines 28--30 after agents complete the token election.
From now, let us consider the configuration after agents complete the token election. 
We first prove that, to execute lines 28--30, agents execute the following procedure.
\begin{enumerate}
\item By executing lines 19--21, an $\Lsleader$ token and an $\Lsleft$ token are generated. 
\item By executing lines 22--24, an $\Leleader$ token and an $\Lsright$ token are generated. 
\item By executing lines 25--27, an $\Lesleader$ token is generated. 
\item Agents execute lines 28--30. 
\end{enumerate}
From now, we show why agents execute the above procedure to execute lines 28--30. 
To execute lines 28--30, an $\Lesleader$ token is required (line 28). 
Recall that, when agents complete the token election, the leader token is $\Lleader$. 
Hence, to generate an $\Lesleader$ token, agents need to execute lines 25--27 (i.e., the item 3 of the procedure is necessary). 
This is because the behavior of lines 25--27 is the only way to generate an $\Lesleader$ token. 
To execute lines 25--27, an $\Leleader$ token is required (line 25). To generate an $\Leleader$ token, agents need to execute lines 22--24 (i.e., the item 2 of the procedure is necessary) because the behavior of lines 22--24 is the only way to generate an $\Leleader$ token. 
Similarly, to execute lines 22--24, an $\Lsleader$ token is required (line 22), and, to generate an $\Lsleader$ token, agents need to execute lines 19--21 (i.e., the item 1 of the procedure is necessary) because the behavior of lines 19--21 is the only way to generate an $\Lsleader$ token. 

In the procedure, agents may perform the behaviors of some items multiple times by resetting the leader token to a $\Lleader$ token (e.g., agents may perform the behaviors of items 1, 2, 3, and 4 after performing the behaviors of items 1 and 2). 
However, we can observe that agents finally execute a procedure such that agents perform the behavior of each item only once in the procedure. 
From now on, we consider only such a procedure. 

From the pseudocode, to execute lines 28--30, the following three conditions should hold during the procedure. 
Note that, after agents complete the token election, $\# \Lright = 1$, $\# \Lleft = 1$, and $\# \Lleader = 1$ hold. 
\begin{itemize}
\item After executing lines 19--21, an agent having an $\Lsleader$ token does not interact with other agents until the agent interacts with an agent having an $\Lright$ token (i.e., the agent interacts only when agents execute lines 22--24). 
Otherwise, agents make the $\Lsleader$ token transition to an $\Lleader$ token and cannot execute lines 22--24 (i.e., the item 2 of the procedure cannot be executed). 
\item After executing lines 19--21, an agent having an $\Lsleft$ token does not interact with other agents until the agent interacts with an agent having an $\Leleader$ token (i.e., the agent interacts only when agents execute lines 25--27).
Otherwise, agents make the $\Lsleft$ token transition to an $\Lleft$ token and cannot execute lines 25--27 (i.e., the item 3 of the procedure cannot be executed).
\item After executing lines 22--24, an agent having an $\Lsright$ token does not interact with other agents until the agent interacts with an agent having an $\Lesleader$ token (i.e., the agent interacts only when agents execute lines 28--30).
Otherwise, agents make the $\Lsright$ token transition to an $\Lright$ token and cannot execute lines 28--30 (i.e., the item 4 of the procedure cannot be executed).
\end{itemize}

From items 1 and 2, an $\Lsleft$ token exists next to an $\Lsleader$ token when agents execute lines 22--24.
Hence, from the pseudocode, an $\Lsright$ token and an $\Lsleft$ token are next to each other just after agents execute lines 22--24. 
In addition, an $\Leleader$ token and the $\Lsright$ token are also next to each other just after agents execute lines 22--24.
To execute lines 25--27, an agent having the $\Leleader$ token must interact with the agent having the $\Lsleft$ token without meeting the agent having the $\Lsright$ token.
Furthermore, the agent having the $\Lsright$ token must not interact with other agents until agents execute lines 28--30.
By the assumption, since agents execute lines 28--30, there are two paths from the agent having the $\Leleader$ token to the agent having the $\Lsleft$ token just after agents execute lines 22--24.
One of the paths is the path via the agent having the $\Lsright$ token.
The other is the path without passing through the agent having the $\Lsright$ token.
Therefore, there is a cycle in $G$. This is a contradiction.
\end{proof}

Next, we show that $\tre$ of the leader token converges to $no$ if there is a cycle on the graph.
\begin{lemma}
\label{lem:cyc}
For any globally-fair execution $\Xi$, if there is a cycle on a given communication graph, $\tre$ of the leader token converges to $no$. 
\end{lemma}
\begin{proof}
Consider a globally-fair execution $\Xi$ with a graph $G$ containing a cycle.
In $\Xi$, let $C$ be a configuration such that $C$ occurs infinitely often. 
From Lemma \ref{lem:onetoken}, eventually agents complete the token election and thus $C$ occurs infinitely often after agents complete the token election. 

Clearly, each condition in lines 2--8 is not satisfied after $C$.
Thus, from the pseudocode, a token moves by any interaction (except for null transitions) after $C$.
This implies that tokens can move freely on $G$ after $C$.
Hence, from global fairness, a configuration $C'$ such that all tokens are on a cycle occurs.
Moreover, there exists a configuration $C''$ such that $C''$ is reachable from $C'$ and $\Lleader$, $\Lleft$, and $\Lright$ tokens are on the cycle in $C''$. 
This is because $C''$ occurs if the following behaviors occur from $C'$. 
\begin{enumerate}
\item Making $\Lleft$ and $\Lright$ tokens: If an agent having an $\Lsleft$ (or $\Lleft$) token and an agent having an $\Lsright$ (or $\Lright$) token can interact in $C'$, they interact and then an $\Lleft$ token and an $\Lright$ token are generated by the behavior of lines 31--43.
Otherwise, since the left token and the right token are on a cycle in $C'$ (and hence an agent with the token has at least two edges), each agent having the token can interact with an agent having no token.
In the case, an agent having an $\Lsleft$ token (resp., an $\Lsright$ token) interacts with an agent having no token, and an $\Lleft$ token (resp., an $\Lright$ token) is generated.
\item 
Making an $\Lleader$ token: If the leader token is an $\Lsleader$ token or an $\Lesleader$ token, an agent having the token interacts with an agent having an $\Lleft$ token or no token. As a result, an $\Lleader$ token is generated.
If the leader token is an $\Leleader$ token, the $\Leleader$ token moves to an agent that is on a cycle and is adjacent to an agent with an $\Lleft$ token (or an $\Lright$ token).
Then, an agent having the $\Leleader$ token interacts with an agent having the $\Lleft$ token (or the $\Lright$ token) and then an $\Lleader$ token is generated.
\end{enumerate}

There exists a configuration such that the configuration is reachable from $C''$ and, on a cycle, an agent having an $\Lleader$ token is adjacent to an agent having an $\Lleft$ token in the configuration.
This is because tokens can move freely on a graph.
In the configuration, agents can execute lines 19--21. If agents execute lines 19--21, the configuration transitions to a configuration such that $\Lsleader$, $\Lsleft$, and $\Lright$ tokens exist in a cycle.
From the configuration, the $\Lright$ token can move to an agent next to an agent with the $\Lsleader$ token while an agent with the $\Lsleader$ token and an agent with the $\Lsleft$ token do not interact with any agent.
This is because they are on a cycle and the $\Lright$ token can move along the cycle.
Then, an agent having the $\Lsleader$ token can interact with an agent having the $\Lright$ token and then agents execute lines 22--24.
Such behavior causes a configuration such that $\Leleader$, $\Lsleft$, and $\Lsright$ tokens are on a cycle.
From the configuration, the $\Leleader$ token can move to an agent next to an agent having the $\Lsleft$ token while an agent having the $\Lsright$ token and an agent having the $\Lsleft$ token do not interact with any agent.
This is because they are on a cycle and the $\Leleader$ token can move along the cycle. Then, an agent having the $\Leleader$ token can interact with an agent having the $\Lsleft$ token and agents can execute lines 25--27.
After that, an agent having an $\Lesleader$ token can interact with an agent having the $\Lsright$ token and agents can execute lines 28--30.
Hence, from global fairness, since each of the configurations occurs infinitely often, agents execute lines 28--30 infinitely often and agents assign $no$ to $\tre$ of the leader token infinitely often.
Although $\tre$ of the leader token transitions to $yes$ if agents execute line 8, agents does not execute line 8 after $C$.
Therefore, the lemma holds. 
\end{proof}

From Lemmas \ref{lem:onetoken}, \ref{lem:nocyc}, and \ref{lem:cyc}, we prove the following theorem.
\begin{theorem}
\label{the:cyc}
Algorithms~\ref{alg:cyc1} and \ref{alg:cyc2} solve the tree identification problem.
That is, there exists a protocol with constant states and designated initial states that solves the tree identification problem under global fairness.
\end{theorem}
\begin{proof}
From Lemma \ref{lem:onetoken}, there is a configuration $C$ such that $\tre$ of the leader token is $yes$ in $C$ and agents complete the token election at $C$.
Hence, from Lemmas \ref{lem:nocyc} and \ref{lem:cyc}, if there is a cycle (resp., no cycle) in a given communication graph, $\tre$ of the leader token converges to $no$ (resp., $yes$).
From the pseudocode, since each token can move freely on the graph, $\tre$ of each agent converges to the same value of $\tre$ of the leader token. 
Thus, if there is a cycle (resp., no cycle) in a given communication graph, $\tre$ of each agent converges to $no$ (resp., $yes$).
Therefore, the theorem holds.
\end{proof}

\subsection{$k$-regular Identification Protocol with knowledge of $P$ under Global Fairness}
In this subsection, we give a $k$-regular identification protocol (hereinafter referred to as ``$k$RI protocol'') with $O(k \log P)$ states and designated initial states under global fairness. 
In this protocol, the upper bound $P$ of the number of agents is given. However, we also show that the protocol solves the problem with $O(k \log n)$ states if the number of agents $n$ is given. 

From now, we explain the basic strategy of the protocol. 
First, agents elect a leader token. In this protocol, agents with leader tokens leave some information in agents. To keep only the information that is left after completion of the election, we introduce \emph{level} of an agent. If an agent at level $i$ has the leader token, we say that the leader token is at level $i$.
Agents with leader tokens leave the information with their levels. Before agents complete the election of leader tokens, agents keep increasing their levels (we explain later how to increase the level), and agents discard the information with smaller levels when agents increase their levels. 
When agents complete the election of leader tokens, the agent with the leader token is the only agent that has the largest level. Then, all agents eventually converge to the level. 
Hence, since agents discard the information with smaller levels, agents virtually discard any information that was left before agents complete the election. 
From now on, we consider configurations after agents elect a leader token and discard any outdated information. 

Now, we explain how the protocol solves the $k$-regular identification problem by using the leader token. 
Concretely, each agent examines whether its degree is at least $k$, and whether its degree is at least $k+1$. If an agent confirms that its degree is at least $k$ but does not confirm that its degree is at least $k+1$, then the agent thinks that its degree is $k$. 
Each agent examines whether its degree is at least $k$ as follows: An agent $a$ with the leader token checks whether $a$ can interact with $k$ different agents. 
To check it, agent $a$ with the leader token marks adjacent agents and counts how many times $a$ has marked. 
Concretely, when agent $a$ having the leader token interacts with an agent $b$, agent $a$ marks agent $b$ by making $b$ change to a marked state. 
Agent $a$ counts how many times $a$ interacts with an agent having a non-marked state (hereinafter referred to as ``a non-marked agent'').
If agent $a$ having the leader token interacts with $k$ non-marked agents successively, $a$ decides that $a$ can interact with $k$ different agents (i.e., its degree is at least $k$).

If an agent confirms that its degree is at least $k$, the agent stores this information locally. 
To do this, we introduce a variable $loc_a$ at agent $a$: 
Variable $loc_a \in \{yes$, $no \}$, initialized to $no$, represents whether the degree of agent $a$ is at least $k$. If $loc_a = yes$ holds, agent $a$ thinks that its degree is at least $k$. 
If an agent $a$ confirms that its degree is at least $k$, agent $a$ stores this information locally by making $loc_a$ transition from $no$ to $yes$. 

Next, we show how agents decide whether the graph is $k$-regular. 
In this protocol, first an agent with the leader token decides whether the graph is $k$-regular, and then the decision is conveyed to all agents by the leader token. 
We use variable $reg_a$ at agent $a$ for the decision: 
Variable $reg_a \in \{yes$, $no \}$, initialized to $no$, represents the decision of the $k$-regular graph. If $reg_a = yes$ holds for agent $a$, then $\gamma(s_a) = yes$ holds.
If $reg_a = no$ holds, then $\gamma(s_a) = no$ holds. 
Whenever an agent $a$ with the leader token makes $loc_a$ transition to $yes$, agent $a$ makes $reg_a$ transition to $yes$. 
If an agent $a$ with the leader token finds an agent $b$ such that $loc_b=no$ or its degree is at least $k+1$, agents reset $reg_a$ to $no$. 
Note that, since all agents follow the decision of the leader token, this behavior practically resets $reg$ of each agent. 
If there is such agent $b$, agent $a$ with the leader token eventually finds agent $b$ since the leader token moves freely on the graph. 
Hence, if the graph is not $k$-regular, $reg$ of the leader token (i.e., $reg_a$ such that agent $a$ has the leader token) transitions to $no$ infinitely often. 
On the other hand, if the graph is $k$-regular, eventually $loc_a$ of each agent $a$ transitions from $no$ to $yes$. Let us consider a configuration where $loc$ of each agent other than an agent $x$ is $yes$ and $loc_x$ is $no$. 
After the configuration, when agent $x$ makes $loc_x$ and $reg_x$ transition to $yes$, agent $x$ has the leader token (i.e., $reg$ of the leader token transitions to $yes$). Hence, since there is no agent such that its $loc$ is $no$ or its degree is at least $k+1$, $reg$ of the leader token never transitions to $no$ afterwards and thus $reg$ of the leader token converges to $yes$. 
Thus, since agents convey the decision of the leader token to all agents, eventually all agents make a correct decision. 

Before we explain the details of the protocol, first we introduce other variables at agent $a$. 
\begin{itemize}
\item $LF_a \in \{L_0$, $L_1$, $\ldots$, $L_{k}$, $\F$, $\F'\}$: 
Variable $LF_a$, initialized to $L_0$, represents states for a leader token and marked agents. 
If $LF_a$ is neither $\F$ nor $\F'$, agent $a$ has a leader token. 
In particular, if $LF_a = L_i (i \in \{0$, $1$, $\ldots$, $k\})$ holds, agent $a$ has an $L_i$ token. Moreover, $LF_a = L_i$ represents that agent $a$ has interacted with $i$ different non-marked agents (i.e., agent $a$ has at least $i$ edges). 
If $LF_a = \F$ holds, agent $a$ has no leader token. If $LF_a = \F'$ holds, agent $a$ has no leader token and $a$ is marked by other agents. 
\item $level_a \in \{0$, $1$, $2$, $\ldots$, $\lfloor \log P \rfloor\}$: Variable $level_a$, initialized to $0$, represents the level of agent $a$. 
\end{itemize}
The protocol uses $O(k \log P)$ states because the number of values taken by variable $LF_a$ is $k+2$, the number of values taken by variable $level_a$ is $\lfloor \log P \rfloor + 1$, and the number of values taken by other variables ($loc_a$ and $reg_a$) is constant.

Now, we explain the details of the protocol. The protocol is given in Algorithm \ref{alg:reg}. 

\begin{algorithm}[t!]
	\caption{A $k$RI protocol}         
	\label{alg:reg}
	\algblockdefx{when}{End}{\textbf{when}}{\textbf{end}}
	\begin{algorithmic}[1]
		\renewcommand{\algorithmicrequire}{\textbf{Variables at an agent $a$:}}
		\Require
		\Statex $LF_a \in \{L_0$, $L_1$, $\ldots$, $L_{k}$, $\F$, $\F'\}$: States for a leader token and marked agents, initialized to $L_0$. 
		\Statex $level_a \in \{0$, $1$, $2$, $\ldots$, $\lfloor \log P \rfloor\}$: States for the level of agent $a$, initialized to $0$. 
      \Statex $loc_a \in \{yes$, $no \}$: States representing whether the degree of agent $a$ is at least $k$, initialized to $no$. 
      \Statex $reg_a \in \{yes$, $no \}$: Decision of the $k$-regular graph, initialized to $no$. 

		\when { agent $a$ interacts with agent $b$} \textbf{do}
		\Statex $\llangle$ The behavior when agents have the same level $\rrangle$ 
		\If{$level_a=level_b$}
		\Statex \{ The election of leader tokens \}
		\If{$LF_a=  L_x$ $\wedge$ $LF_b = L_y$ ($x, y \in \{0$, $1$, $2$, $\ldots$, $k\}$)}
		\State $level_a \leftarrow level_a + 1$
		\State $LF_a \leftarrow L_0$, $LF_b \leftarrow \F$
		\State $reg_a \leftarrow no$
		\State $loc_a \leftarrow no$
		\Statex \{ Decision and movement of the token \}
		\ElsIf{$LF_a = L_x $ $\wedge$ $LF_b = \F$ ($x \in \{0$, $1$, $2$, $\ldots$, $k-2\}$)}
		\State $LF_a \leftarrow L_{x+1}$, $LF_b \leftarrow \F'$
		\ElsIf{$LF_a = L_x $ $\wedge$ $LF_b = \F'$ ($x \in \{0$, $1$, $2$, $\ldots$, $k\}$)}
		\State $LF_a \leftarrow \F$, $LF_b \leftarrow L_0$
		\State $reg_b \leftarrow reg_a$
		\ElsIf{$LF_a = L_{k-1}$ $\wedge$ $LF_b = \F$}
		\State $LF_a \leftarrow L_{k}$, $LF_b \leftarrow \F'$
		\If{$loc_a = no$}
		\State $reg_a \leftarrow yes$
		\State $loc_a \leftarrow yes$
		\EndIf
		\Statex \{ Reset of $reg$ of the leader token (the degree of agent $a$ is at least $k+1$) \}
		\ElsIf{$LF_a = L_{k}$ $\wedge$ $LF_b = \F$}
		\State $LF_a \leftarrow L_0$, $LF_b \leftarrow \F'$
		\State $reg_a \leftarrow no$
		\EndIf
		\Statex \{ Reset of $reg$ of the leader token ($loc_a$ or $loc_b$ is $no$) \}
		\If{$loc_a = no$ $\lor$ $loc_b = no$}
		\State $reg_a \leftarrow no$, $reg_b \leftarrow no$
		\EndIf
		\Statex $\llangle$ The behavior when agents have different levels $\rrangle$ 
		\ElsIf{$level_a>level_b$} 
		\State $level_b \leftarrow level_a$
		\State $loc_b \leftarrow no$
		\State $LF_b \leftarrow \F$
		\EndIf
		\End
	\end{algorithmic}
\end{algorithm}

\paragraph*{The election of leader tokens with levels (lines 2--7 and 26--30 of the pseudocode)} 
Initially, each agent has the leader token and the level of each agent is 0. 
If two agents with leader tokens at the same level interact, agents delete one of the leader tokens and increase the level of the agent with the remaining leader token by one (lines 2--5). 
Moreover, $loc_a$ and $reg_a$ transition to $no$ (lines 6--7), where agent $a$ is the agent with the remaining leader token. 
Next, we consider the case where two agents at different levels interact. 
If an agent $a$ at the larger level interacts with an agent $b$ at the smaller level, agent $b$ update its level to the same level as the larger level (regardless of possession of the leader token). This behavior appears in lines 26--27. Furthermore, at the interaction, agent $b$ resets $loc_b$ to $no$ (line 28), and agent $b$ deletes its leader token if agent $b$ has the leader token (line 29). 
We can observe that there is level $lev\_last$ such that all agents converge to level $lev\_last$, because agents update their levels by only above behaviors and there is no behavior that increases the number of leader tokens. 
Since an agent at the largest level updates its level only if the agent has the leader token, there is an agent with the leader token at the largest level in any configuration. 
Thus, since each agent converges to level $lev\_last$ and the leader token moves freely among agents at the same level (we will show this movement behavior later), eventually agents elect a leader token by above behaviors. 

Agents at the largest level delete the leader token only by the behavior of lines 3--7. 
This implies that, if at least two agents at the largest level have the leader token, eventually agents at the largest level with the leader tokens interact and then the largest level is updated. 
Hence, only one leader token can obtain level $lev\_last$. 
When an agent $a$ with the leader token updates its level to level $lev\_last$ by the behavior of lines 3--7, agent $a$ resets $loc_a$ and $reg_a$ to $no$. 
Since other agents are at levels smaller than level $lev\_last$ just after the interaction, other agents will reset their $loc$ to $no$ by the behavior of line 28. 
From these facts, the interaction causes a configuration such that 1) the number of agents with the leader token at level $lev\_last$ is one, 2) $reg_a=no$ and $loc_a=no$ hold for the agent $a$ with the leader token at level $lev\_last$, and 3) other agents will reset their $loc$ after the configuration. 

Note that, since agents delete one leader token by the behavior of lines 3--7, at most half of leader tokens at level $i$ update their level to $i+1$ for $0 \le i$. 
Thus, since there is no behavior that increases the number of leader tokens, the maximum level is at most $\lfloor \log n \rfloor$. 
In this protocol, since only the upper bound $P$ of the number of agents is given, the maximum level is at most $\lfloor \log P \rfloor$. 

\paragraph*{Search for an agent whose degree is at least $k$ or at least $k+1$ with levels (lines 8--18 of the pseudocode)} 

First of all, this search behavior is performed only on the same level (line 2). 
Recall that eventually all agents converge to the same level (and agents discard the information at other levels).

In this behavior, to examine degrees of agents, agents use the leader token. 
An agent having the leader token confirms whether the agent can interact with $k$ different agents, so that the agent confirms that its degree is at least $k$. 
To confirm it, the agent marks adjacent agents one by one and counts how many times the agent interacts with a non-marked agent. 
Concretely, when an agent $a$ having the $L_i$ token interacts with an agent $b$ having $\F$, agent $a$ marks agent $b$ (i.e., $LF_b$ transitions to $\F'$). 
At the interaction, $a$ makes the $L_i$ token transition to the $L_{i+1}$ token. 
These behaviors appear in lines 8--9. 
If agent $a$ obtains the $L_{j}$ token by such an interaction, agent $a$ has interacted with $j$ different agents because $a$ has marked $j$ non-marked agents. 
Thus, when agent $a$ having the $L_{k-1}$ token interacts with an agent with $\F$, agent $a$ notices that $a$ has at least $k$ edges and thus $a$ updates $loc_a$ and $reg_a$ to $yes$ (lines 13--18). 
Similarly, when agent $a$ having the $L_{k}$ token interacts with an agent with $\F$, agent $a$ notices that $a$ has at least $k+1$ edges and thus $a$ updates $reg_a$ to $no$ (lines 19--22). 

When agent $a$ having the $L_i$ token interacts with a marked agent $b$ (i.e., agent $b$ with $\F'$), agent $a$ resets the $L_i$ token to the $L_0$ token. 
Moreover, at the interaction, $a$ deletes a mark of $b$ and carries the $L_0$ token to $b$ (i.e., $LF_b = \F'$ transitions to $L_0$ and $LF_a$ transitions to $\F$). 
These behaviors appear in lines 10--12. 
By these behaviors, the leader token can move freely on a graph because an agent having the leader token can mark any adjacent agent by an interaction. 
Note that, after the leader token moves, some agents may remain as marked agents. However, even in the case, eventually agents correctly detect an agent whose degree is at least $k$ or $k+1$ because an agent with the leader token can delete marks of adjacent agents freely. 
Concretely, an agent $a$ having the leader token deletes a mark of the adjacent agent $b$ by making interaction between $a$ and $b$ three times. 
Figure \ref{fig:thr} shows the example of the three interactions. 
By the interactions, the agents carry the leader token from $a$ to $b$ and then agent $b$ returns the leader token to $a$.
As a result, agent $a$ has the leader token again and the marked agent $b$ transitions to a non-marked agent. 

By the above behaviors, eventually each agent with degree $k$ makes its $loc$ transition to $yes$. 
If agents find some agent $a$ with $loc_a=no$, agents make $reg$ of the leader token transition to $no$ (lines 23--25). 
Thus, if there is an agent whose degree is not $k$, $reg$ of the leader token converges to $no$. 
On the other hand, if the degree of each agent is $k$, eventually each agent makes its $loc$ transition to $yes$, and there is no behavior that makes $reg$ of the leader token transition to $no$ afterwards. 
Hence, since $loc$ and $reg$ of the leader token transition to $yes$ simultaneously by lines 15--18, $reg$ of the leader token converges to $yes$ in the case. 
When agents move the leader token, agents convey $reg$ of the leader token (line 12). 
Therefore, eventually each agent makes a correct decision. 

\begin{figure}[t!]
\begin{center}
\includegraphics[scale=0.5]{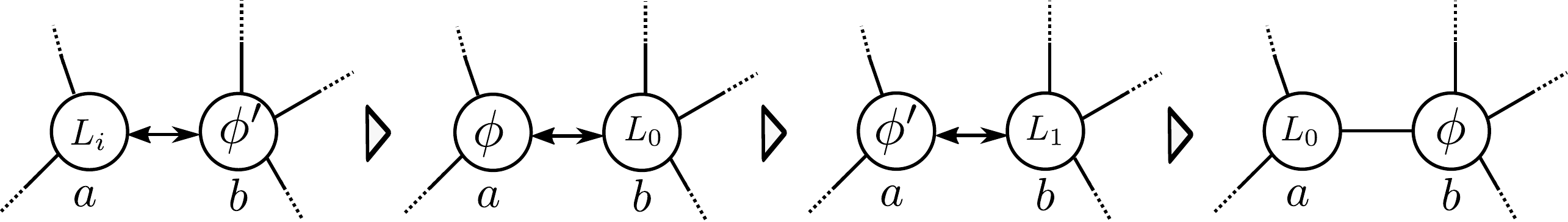}
\caption{The example of deleting a mark}
\label{fig:thr}
\end{center}
\end{figure}

\subsection*{Correctness}
First of all, we define some notations for the level. 
Let $lev(a, C)$ be level of agent $a$ in a configuration $C$. 
For a set of all agents $V=\{v_1$, $v_2$, $\ldots$, $v_n\}$ and a configuration $C$, let $lev\_max(C) = \max_{a \in V} \{ lev(a, C) \}$.  

To begin with, we show that, in any execution of Algorithm \ref{alg:reg}, the behavior of line 4 is not performed by agent $a$ such that $level_a = \lfloor \log n \rfloor$ holds. 
This implies that the domain of variable $level$ is valid. 
Note that each agent $a$ increases its $level_a$ one by one. 

\begin{lemma}
\label{lem:reg:maxlevel}
In any execution, agent $a$ does not increase $level_a$ if $level_a = \lfloor \log n \rfloor$ holds.
\end{lemma}
\begin{proof}
First of all, from the pseudocode, there is no behavior that decreases the number of leader tokens, and there is no behavior that decreases the level of an agent. 

From now, we show that, in any execution, at most one leader token at level $\lfloor \log n \rfloor$ can occur. 
From the pseudocode, the number of leader tokens at some level increases only if the behavior of lines 3--7 occurs. 
Hence, by the behavior of lines 3--7, if the number of leader tokens at level $i$ increases by one, agents delete two leader tokens at level $i-1$. 
Thus, since an initial level of each agent is 0 and the initial number of leader tokens is $n$, at most one leader token at level $\lfloor \log n \rfloor$ can occur. 

From the pseudocode, the behavior of line 4 occurs only if two agents with leader tokens at the same level interact. 
Hence, if $level_a = \lfloor \log n \rfloor$ holds, agent $a$ cannot increase $level_a$ and thus the lemma holds. 
\end{proof}

From now, we show that Algorithm \ref{alg:reg} solves the problem. 
First, we prove that, in any configuration $C$ of any execution, there exists an agent with the leader token at level $lev\_max(C)$. 

\begin{lemma}
\label{lem:reg:maxle}
Let us consider a graph $G=(V, E)$, where $V=\{v_1$, $v_2$, $\ldots$, $v_n\}$. 
For any configuration $C$ of any execution $\Xi$ with $G$, there exists an agent $v_m$ with the leader token such that $lev(v_m, C) = lev\_max(C)$ holds. 
\end{lemma}
\begin{proof}
For any graph $G=(V,E)$, let us consider an execution $\Xi=C_0$, $C_1$, $C_2$, $\ldots$ of the protocol, where $V=\{v_1$, $v_2$, $\ldots$, $v_n\}$.
We prove the lemma by induction on the index of a configuration. 
In the base case ($C_0$), clearly $level_{v_1} = level_{v_2} = \cdots = level_{v_n} = lev\_max(C_0) =0$ holds.
For the induction step, we assume that, there exists an agent $v_m$ with the leader token such that $lev(v_m, C_k) = lev\_max(C_k)$. 
Let us consider an interaction at $C_k \rightarrow C_{k+1}$ for four cases. 
\begin{itemize}
\item Case where the behavior of lines 3--7, 10--12, or 26--30 does not occur at $C_k \rightarrow C_{k+1}$: From the pseudocode, by the interaction, agents do not update $level_{v_i}$ for $1 \le i \le n$. In addition, by the interaction, agents do not move the leader token, and the number of leader tokens is not changed. 
Hence, in this case, agent $v_m$ with the leader token satisfies that $lev(v_m, C_{k+1}) = lev\_max(C_{k+1})$. 
\item Case where the behavior of lines 3--7 occurs at $C_k \rightarrow C_{k+1}$: From the pseudocode, if an agent increases its level by the interaction, the agent has the leader token after the interaction. 
Hence, if $lev\_max(C_{k+1}) = lev\_max(C_k)$ holds, agent $v_m$ does not join the interaction and thus $lev(v_m, C_{k+1}) = lev\_max(C_{k+1})$ holds and agent $v_m$ has the leader token in $C_{k+1}$. 
On the other hand, if $lev\_max(C_{k+1}) > lev\_max(C_k)$ holds, one of the interacting agents $v'_m$ with the leader token satisfies $lev(v_{m'}, C_{k+1})$ $=$ $lev\_max(C_{k+1})$. 
\item Case where the behavior of lines 10--12 occurs at $C_k \rightarrow C_{k+1}$: From the pseudocode, by the interaction, agents move the leader token from an interacting agent to the other interacting agent. The interacting agents have the same level before the interaction. Moreover, by the interaction, agents do not change their level and thus $lev\_max(C_{k+1}) = lev\_max(C_k)$ holds. 
Hence, if agent $v_m$ does not join the interaction, $lev(v_m, C_{k+1}) = lev\_max(C_{k+1})$ holds and agent $v_m$ has the leader token in $C_{k+1}$. 
If agent $v_m$ joins the interaction, other interacting agent $v_{m'}$ with the leader token satisfies $lev(v_{m'}, C_{k+1})$ $=$ $lev\_max(C_{k+1})$. 
\item Case where the behavior of lines 26--30 occurs at $C_k \rightarrow C_{k+1}$: From the pseudocode, by the interaction, an interacting agent $a$ at the larger level do not change its variables, and the other interacting agent $b$ at the smaller level becomes the same level as the level of agent $a$. 
Hence, in this case, agent $v_m$ with the leader token satisfies that $lev(v_m, C_{k+1}) = lev\_max(C_{k+1})$ (whether $v_m$ joins the interaction or not). 
\end{itemize}
In each case, there exists an agent $v_m$ with the leader token such that $lev(v_m, C_{k+1}) = lev\_max(C_{k+1})$. 
Therefore, the lemma holds. 
\end{proof}

Next, we show that, in any execution, the level of each agent converges to the same value. 
\begin{lemma}
\label{lem:reg:level}
For any execution $\Xi$ with some graph $G$, there exist a configuration $C$ and $lev\_last$ $(0 \le lev\_last \le \lfloor \log n \rfloor)$ such that $level_a= lev\_last$ holds for each agent $a$ after $C$.
\end{lemma}
\begin{proof}
Let us consider a graph $G=(V, E)$ and an execution $\Xi=C_0$, $C_1$, $\ldots$ with $G$, where $V=\{v_1$, $v_2$, $\ldots$, $v_n\}$. 
From the pseudocode, since there is no behavior that increases the number of leader tokens, the number of leader tokens does not change after some configuration $C_i$. 
This implies that, after $C_i$, the behavior of lines 3--7 does not occur. Thus, no agent can obtain the level larger than $lev\_max(C_i)$. 
Let us consider two agents $v_x$ and $v_y$ that interact at $C_j \rightarrow C_{j+1}$ for $j \ge i$. 
If $lev(v_x, C_j)>lev(v_y, C_j)$ holds, $lev(v_y, C_{j+1})=lev(v_x, C_{j+1})=lev(v_x, C_j)$ holds by the behavior of lines 26--30. 
Hence, since there is no behavior that decreases level of an agent, $lev(v_1, C_k) = lev(v_2, C_k) = \cdots = lev(v_n, C_k) = lev\_max(C_i)$ holds for a configuration $C_k$ $(k \ge i)$ from global fairness. Since each agent maintains $lev\_max(C_i)$ after $C_k$, the lemma holds. 
\end{proof}

Let us consider some execution $\Xi^*$ with some graph $G$. 
Let $lev\_last$ be level such that $level_a = lev\_last$ holds for each agent $a$ after some configuration of $\Xi^*$. 
From Lemma \ref{lem:reg:level}, such $lev\_last$ exists. 
From now, we show some properties about $\Xi^*$ and $lev\_last$. 
First, we prove that, in $\Xi^*$, after some of agents obtains level $lev\_last$, there is only one leader token that is at level $lev\_last$.
\begin{lemma}
\label{lem:reg:le1}
In $\Xi^*$, after some of agents obtains level $lev\_last$, there is only one leader token that is at level $lev\_last$.
\end{lemma}
\begin{proof}
Let $C$ be a configuration in $\Xi^*$ such that the first agent at level $lev\_last$ appears (i.e., there is an agent at level $lev\_last$ in $C$ and there is no agent that is at level $lev\_last$ before $C$). 
We show that there is only one leader token that is at level $lev\_last$ after $C$. 
First of all, from Lemma \ref{lem:reg:maxle}, after $C$, there is at least one leader token that is at level $lev\_last$. 
Thus, for the purpose of contradiction, we assume that there exists a configuration $C'$ of $\Xi^*$ such that there are two or more leader tokens that are at level $lev\_last$ in $C'$ and $C'$ occurs after $C$. 
Let $\omega_1$ and $\omega_2$ be the leader tokens in $C'$. 
Let us consider a configuration $C''$ such that $level_a = lev\_last$ holds for each agent $a$ after $C''$ and $C''$ occurs after $C'$. If agents delete $\omega_1$ or $\omega_2$, the behavior of lines 3--7 or 26--30 occurs. However, both behaviors must not occur from the definition of $lev\_last$. 
Hence, there also exist $\omega_1$ and $\omega_2$ in $C''$. 

After $C''$, when an agent having a leader token interacts with an agent having no leader token, agents move the leader token, or, by making an additional interaction between them, agents move the leader token. This implies that $\omega_1$ and $\omega_2$ can move to any agent after $C''$. 
Hence, from global fairness, eventually an agent having $\omega_1$ interacts with an agent having $\omega_2$ and then they update their levels to $lev\_last + 1$ by the behavior of lines 3--7.
This contradicts the definition of $lev\_last$. 
\end{proof}

In addition, since agents at level $lev\_last$ do not delete the leader token, Lemma \ref{lem:reg:le1} can be extended as follow. 

\begin{lemma}
\label{lem:reg:le1:2}
In $\Xi^*$, only one leader token can be at level $lev\_last$. 
\end{lemma}



Next, we show properties about an agent at level $lev\_last$ whose degree is less than $k+1$. 
\begin{lemma}
\label{lem:reg:degk-1}
In $\Xi^*$, an agent $a$ at level $lev\_last$ does not perform the behavior of lines 15--18 if the degree of agent $a$ is less than $k$. 
Moreover, in $\Xi^*$, an agent $b$ at level $lev\_last$ does not perform the behavior of lines 19--22 if the degree of agent $b$ is less than $k+1$.
\end{lemma}
\begin{proof}
In $\Xi^*$, when an agent $c$ with the leader token updates its level from $lev\_last -1$ to $lev\_last$ by the behavior of lines 3--7, $LF_c$ transitions to $L_0$. 
Note that, from Lemma \ref{lem:reg:le1:2}, only agent $c$ with the leader token at level $lev\_last$ updates its level from $lev\_last -1$ to $lev\_last$ by the behavior of lines 3--7 in $\Xi^*$. 
Hence, when other agents update its level to $lev\_last$, the agents reset their $LF$ to $\F$. 

From now, we consider interactions between agents at level $lev\_last$. 
To execute lines 15--18, it needs to occur $k$ times that an agent having the leader token interacts with a non-marked agent without moving the leader token. 
This is because the leader token transitions to $L_0$ when it moves. 
When an agent having the leader token interacts with a non-marked agent, the agent with the leader token marks the non-marked agent. 
From the pseudocode, unless the leader token moves, a marked agent at level $lev\_last$ never transitions to a non-marked agent. 
From these facts, since there is only one leader token, agents never execute lines 15--18 if there is no agent whose degree is at least $k$. 
Similarly, agents never execute lines 19--22 if there is no agent whose degree is at least $k+1$. 
Therefore, the lemma holds. 
\end{proof}

Now, we show properties about an agent whose degree is at least $k$. 
We prove that, if there is an agent whose degree is at least $k$ (resp., $k+1$), the agent performs the behavior of lines 13--18 (resp., lines 19--22) infinitely often. 
\begin{lemma}
\label{lem:reg:degk}
In $\Xi^*$, an agent $a$ performs the behavior of lines 13--18 infinitely often if the degree of agent $a$ is at least $k$. 
Moreover, in $\Xi^*$, an agent $b$ performs the behavior of lines 19--22 infinitely often if the degree of agent $b$ is at least $k+1$. 
\end{lemma}
\begin{proof}
From Lemmas \ref{lem:reg:level} and \ref{lem:reg:le1:2}, there exists a configuration $C$ in $\Xi^*$ such that there is only one leader token and each agent is at level $lev\_last$ after $C$. 
We consider configurations after $C$. 
From the pseudocode, the leader token can move to any agent after $C$. 
Thus, from global fairness, there exists a configuration $C'$ such that $C'$ occurs infinitely often and agent $a$ whose degree is at least $k$ has the leader token in $C'$.

In $C'$, since there exists only one leader token, each agent adjacent to $a$ has $\F$ or $\F'$. 
To make all $\F'$ adjacent to $a$ transition to $\F$, we consider the following procedure.
\begin{enumerate}
\item Agent $a$ having the leader token interacts with an agent $c$ having $\F'$ three times. 
From the pseudocode, after the interactions, agent $a$ has the $L_0$ token and agent $c$ has $\F$ (and other agents are the same states as before the interactions). 
\item Agents repeat the above behavior until no agent adjacent to $a$ has $\F'$. 
\end{enumerate}
From global fairness, eventually agents perform the above procedure. 
Hence, there exists a configuration such that, for agent $a$ with the leader token, each agent adjacent to $a$ has $\F$ and the configuration occurs infinitely often.  
Moreover, from the configuration, eventually agent $a$ interacts with each adjacent agent in a row. 
From the pseudocode, when such interactions occur, agents execute lines 13--18.  
Thus, if the degree of agent $a$ is at least $k$, agent $a$ performs the behavior of lines 13--18 infinitely often. 
Similarly, if the degree of agent $b$ is at least $k+1$, agent $b$ performs the behavior of lines 19--22 infinitely often. 
\end{proof}

From Lemmas \ref{lem:reg:level} and \ref{lem:reg:le1:2}, eventually agents complete the election of leader tokens in $\Xi^*$.
From now on, we define $reg$ of the leader token as $reg_a$ such that agent $a$ has the elected leader token. 
We show that $reg$ of the leader token transitions to $yes$ only a finite number of times in $\Xi^*$. 
\begin{lemma}
\label{lem:reg:finite}
In $\Xi^*$, $reg$ of the leader token transitions to $yes$ only a finite number of times. 
\end{lemma}
\begin{proof}
From the pseudocode, $reg$ of the leader token transitions to $yes$ only by the behavior of lines 15--18. 
Hence, we prove that this behavior occurs only a finite number of times. 
From Lemmas \ref{lem:reg:level} and \ref{lem:reg:le1:2}, there exists a configuration $C$ in $\Xi^*$ such that there is only one leader token and each agent is at level $lev\_last$. 
From the pseudocode, $loc_a$ of an agent $a$ transitions from $yes$ to $no$ only by the behaviors of lines 3--7 and 26--30. 
The behavior of lines 3--7 occurs only if there are multiple leader tokens, and the behavior of lines 26--30 occurs only if there are agents at different levels. 
From these facts, $loc_a$ of each agent $a$ transitions from $yes$ to $no$ a finite number of times in $\Xi^*$. 
Hence, from the condition of line 15, the behavior of lines 15--18 occurs a finite number of times in $\Xi^*$. 
Therefore, the lemma holds. 
\end{proof}

By using the above lemmas, we show that the protocol solves the $k$-regular identification problem. 
First, we show that, if the given graph is not $k$-regular, $reg_a$ of each agent $a$ converges to $no$. 

\begin{lemma}
\label{lem:reg:no}
If the given communication graph is not $k$-regular, $reg_a$ of each agent $a$ converges to $no$ in any globally-fair execution. 
\end{lemma}
\begin{proof}
Let $\Xi$ be an execution with a non $k$-regular graph. 
Let $lev\_last$ be level such that the level of each agent converges to level $lev\_last$ in $\Xi$. 
From Lemma \ref{lem:reg:level}, such $lev\_last$ exists. 

First, we show that, if there is an agent $a$ whose degree is less than $k$ on the graph, $reg$ of the leader token converges to $no$ in $\Xi$. 
From Lemma \ref{lem:reg:degk-1}, if there is an agent $a$ whose degree is less than $k$ on the graph, agent $a$ never updates its $loc_a$ from $no$ to $yes$ after agent $a$ obtains level $lev\_last$. 
When agent $a$ obtains level $lev\_last$, agent $a$ updates its $loc_a$ to $no$ by the behavior of lines 3--7 or 26--30. 
Hence, since $loc_a$ converges to $no$, $reg$ of the leader token transitions to $no$ infinitely often by the behavior of lines 23--25. 
From Lemma \ref{lem:reg:finite}, $reg$ of the leader token transitions from $no$ to $yes$ only a finite number of times in $\Xi$. 
Thus, if there is an agent whose degree is less than $k$ on the graph, $reg$ of the leader token converges to $no$ in $\Xi$. 

Next, we show that, if there is an agent whose degree is at least $k+1$ on the graph, $reg$ of the leader token converges to $no$ in $\Xi$. 
From Lemma \ref{lem:reg:degk}, if there is an agent whose degree is at least $k+1$ on the graph, $reg$ of the leader token transitions to $no$ infinitely often in $\Xi$. 
Thus, if there is an agent whose degree is at least $k+1$ on the graph, $reg$ of the leader token converges to $no$ in $\Xi$. 

Now, we show that $reg_a$ of each agent $a$ converges to $no$ in $\Xi$. 
Let $C$ be a configuration in $\Xi$ such that there is only one leader token and each agent is at level $lev\_last$ after $C$. From Lemmas \ref{lem:reg:level} and \ref{lem:reg:le1:2}, such $C$ exists in $\Xi$.
From the pseudocode, after $C$, the leader token moves freely on the graph and the decision of the leader token is conveyed to all agents. Thus, since $reg$ of the leader token converges to $no$ in $\Xi$, $reg_a$ of each agent $a$ converges to $no$ in $\Xi$ from global fairness. 
\end{proof}

Next, we show that, if the given graph is $k$-regular, $reg_a$ of each agent $a$ converges to $yes$. 

\begin{lemma}
\label{lem:reg:yes}
If the given communication graph is $k$-regular, $reg_a$ of each agent $a$ converges to $yes$ in any globally-fair execution. 
\end{lemma}
\begin{proof}
Let $\Xi$ be an execution with a $k$-regular graph. 
In $\Xi$, let us consider a configuration $C$ such that 1) there is only one leader token, 2) each agent has some level $lev\_last$ after $C$, and 3) $loc_a$ of some agent $a$ is $no$ in $C$. 
From the pseudocode, when an agent $b$ increases its level, agent $b$ updates its $loc_b$ to $no$. 
Thus, from Lemmas \ref{lem:reg:level} and \ref{lem:reg:le1:2}, such $C$ exists in $\Xi$. 

From Lemma \ref{lem:reg:degk}, each agent performs the behavior of lines 13--18 infinitely often after $C$. Hence, eventually each agent $a$, such that $loc_a =no$ in $C$, makes $loc_a$ transition to $yes$ by the behavior of lines 15--18 after $C$. 
Since there is only one leader token and each agent is at level $lev\_last$ after $C$, the behaviors of lines 3--7 and 26--30 do not occur after $C$. 
This implies that, after $C$, $loc_a$ of each agent $a$ keeps $yes$ if $loc_a=yes$. 

Let us consider the last agent that performs the behavior of lines 15--18. From Lemma \ref{lem:reg:finite}, since $reg$ of the leader token transitions to $yes$ by the behavior of lines 15--18, the behavior of lines 15--18 occurs a finite number of times and thus such an agent exists. 
From the pseudocode, when the agent performs the behavior, $reg$ of the leader token transitions to $yes$. After that, since $loc_a$ of each agent $a$ is $yes$, the behavior of lines 23--25 does not occur. Furthermore, from Lemma \ref{lem:reg:degk-1}, the behavior of lines 19--22 is not performed. 
Thus, $reg$ of the leader token does not transition to $no$ afterwards and thus $reg_a$ of each agent $a$ converges to $yes$. 
\end{proof}

From Lemma \ref{lem:reg:no}, if the given communication graph is not $k$-regular, $reg_a$ of each agent $a$ converges to $no$ in any execution of the protocol. 
From Lemma \ref{lem:reg:yes}, if the given communication graph is $k$-regular, $reg_a$ of each agent $a$ converges to $yes$ in any execution of the protocol. 
Thus, we can obtain the following theorem. 
\begin{theorem}
Algorithm~\ref{alg:reg} solves the $k$-regular identification problem. 
That is, if the upper bound $P$ of the number of agents is given, there exists a protocol with $O(k \log P)$ states and designated initial states that solves the $k$-regular identification problem under global fairness. 
\end{theorem}

Clearly, when the number of agents $n$ is given, the protocol works even if the protocol uses variable $level_a = \{0$, $1$, $2$, $\ldots$, $\lfloor \log n \rfloor\}$ instead of $level_a = \{0$, $1$, $2$, $\ldots$, $\lfloor \log P \rfloor\}$. 
Therefore, we have the following theorem.
\begin{theorem}
If the number of agents $n$ is given, there exists a protocol with $O(k \log n)$ states and designated initial states that solves the $k$-regular identification problem under global fairness. 
\end{theorem}

\subsection{Star Identification Protocol with initial knowledge of $n$ under weak fairness}
In this subsection, we give a star identification protocol (hereinafter referred to as ``SI protocol'') with $O(n)$ states and designated initial states under weak fairness.
In this protocol, the number of agents $n$ is given. 
Recall that, in this protocol under weak fairness, if a transition $(p,q) \rightarrow (p',q')$ exists for $p \neq q$, a transition $(q,p) \rightarrow (q',p')$ also exists.
Since a given graph is a star if $n \le 2$ holds, we consider the case where $n$ is at least 3. 

The basic strategy of the protocol is as follows. 
Initially, each agent thinks that the given graph is not a star.
First, agents elect an agent with degree two or more as a central agent (i.e., an agent that connects to all other agents in the star graph). 
Then, by counting the number of agents adjacent to the central agent, agents examine whether there is a star subgraph in the given graph such that the subgraph consists of $n$ agents. Concretely, if the central agent confirms by counting that there are $n-1$ adjacent agents, agents confirm that there is the subgraph. 
In this case, agents think that the given graph is a star. Then, if two agents other than the central agent interact, agents decide that the graph is not a star. If such an interaction does not occur, agents continue to think that the given graph is a star. 



To explain the details, first we introduce variables at an agent $a$. 
\begin{itemize}
\item $LF_a \in \{\F$, $\F'$, $l'$, $L_2$, $L_3$, $\ldots$, $L_{n-1}\}$: Variable $LF_a$, initialized to $\F$, represents a role of agent $a$. $LF_a=L_i$ means that a central agent $a$ has marked $i$ agents (i.e., agent $a$ has at least $i$ edges). $LF_a=l'$ means that $a$ is a candidate of a central agent and is a marked agent. $LF_a=F$ means that agent $a$ is a non-marked agent. $LF_a=F'$ means that agent $a$ is a marked agent. 
When $LF_a=x$ holds, we refer to $a$ as an $x$-agent.
\item $star_a \in \{yes$, $no$, $never \}$: Variable $star_a$, initialized to $no$, represents a decision of a star. 
If $star_a = yes$ holds, $\gamma(s_a) = yes$ holds (i.e., $a$ decides that a given graph is a star).
If $star_a = no$ or $star_a=never$ holds, $\gamma(s_a) = no$ holds (i.e., $a$ decides that a given graph is not a star). 
$star_a=never$ means the stronger decision of $no$. If agent $a$ with $star_a=never$ interacts with agent $b$, $star_b$ transitions to $never$ regardless of the value of $star_b$.  
\end{itemize}
The protocol is given in Algorithm \ref{alg:starn}. Algorithm~\ref{alg:starn} uses $3n+3$ states because the number of values taken by variable $LF_a$ is $n+1$ and the number of values taken by variable $star_a$ is 3.

\begin{algorithm}[t!]
	\caption{A SI protocol}         
	\label{alg:starn}
	\algblockdefx{when}{End}{\textbf{when}}{\textbf{end}}
	\begin{algorithmic}[1]
		\renewcommand{\algorithmicrequire}{\textbf{A variable at an agent $a$:}}
		\Require
		\Statex $LF_a \in \{\F$, $\F'$, $l'$, $L_2$, $L_3$, $\ldots$, $L_{n-1}\}$: States that represent roles of agents, initialized to $\F$.  $L_i$ represents a central agent, $l'$ represents a candidate of the central agent, $\F'$ represents a marked agent, and $\F$ represents a non-marked agent. 
       \Statex $star_a \in \{yes$, $no$, $never \}$: Decision of a star, initialized to $no$. 
		\when { agent $a$ interacts with agent $b$} \textbf{do}
		\Statex $\llangle$ The behavior when $star_a$ or $star_b$ is $never$ $\rrangle$
		\If{$star_a= never$ $\lor$ $star_b = never$}
		\State $star_a \leftarrow never$, $star_b \leftarrow never$
		\Statex $\llangle$ The behaviors when $star_a \neq never$ and $star_b \neq never$ holds $\rrangle$
		\Else
		\Statex \{ The election of a central agent \}
		\If{$LF_a= \F$ $\wedge$ $LF_b = \F$}
		\State $LF_a \leftarrow l'$, $LF_b \leftarrow l'$
		\ElsIf{$LF_a= l'$ $\wedge$ $LF_b = \F$}
		\State $LF_a \leftarrow L_2$, $LF_b \leftarrow \F'$
		\Statex \{ Counting the number of adjacent agents by the central agent \}
		\ElsIf{$LF_a = L_i $ $\wedge$ $LF_b = \F$ ($2 \le i \le n-2$)}
		\State $LF_a \leftarrow L_{i+1}$, $LF_b \leftarrow \F'$
		\EndIf
		\If{$LF_a= L_{n-1}$} 
		\State $star_a \leftarrow yes$, $star_b \leftarrow yes$
		\EndIf
		\Statex \{ Decision of $never$ \}
		\If{$LF_a= \F'$ $\wedge$ $LF_b = \F'$}
		\State $star_a \leftarrow never$, $star_b \leftarrow never$
		\ElsIf{$LF_a= \F'$ $\wedge$ $LF_b = l'$}
		\State $star_a \leftarrow never$, $star_b \leftarrow never$
		\EndIf
		\Statex \{ Conveyance of $yes$ \}
		\If{$star_a= yes$ $\lor$ $star_b = yes$}
		\State $star_a \leftarrow yes$, $star_b \leftarrow yes$
		\EndIf
		\EndIf
		\End
	\end{algorithmic}
\end{algorithm}

Now, we show the details of the protocol. 
\paragraph*{Examination of the subgraph (lines 5--14 and 20--22 of the pseudocode)}
First, agents elect a central agent. Initially, $LF_a$ of each agent $a$ is $\F$. 
When two $\F$-agents interact, both $\F$-agents transition to $l'$-agents (lines 5--6). When an $l'$-agent interacts with an $\F$-agent, the $l'$-agent transitions to an $L_2$-agent (lines 7--8). By these behaviors, we can observe that the $L_2$-agent is adjacent to at least two agents. 
Clearly, if the given graph is a star, agents correctly elect the central agent. 
If the given graph is not a star, agents may elect no central agent or multiple central agents. 
However, in both cases, $star$ of each agent keeps $no$ (we explain the details later). 

Then, to confirm that the central agent is adjacent to $n-1$ agents, the central agent marks adjacent agents one by one and counts how many times the agent interacts with a non-marked agent. 
Concretely, for $2 \le i \le n-3$, when the $L_i$-agent interacts with an $\F$-agent, the $L_i$-agent marks the $\F$-agent (i.e., the $\F$-agent transitions to an $\F'$-agent), and the $L_i$-agent transitions to the $L_{i+1}$-agent (lines 9--11). 
Note that, when an agent becomes the central agent, the agent has already interacted two different agents. Thus, in this protocol, the agent marked the two agents by the interactions, and the agent starts as an $L_2$-agent (lines 5--8). Recall that an $l'$-agent is a marked agent. 

When an agent becomes an $L_{n-1}$-agent, the $L_{n-1}$-agent notices that the agent is adjacent to $n-1$ agents. 
Thus, agents notice that there is a star subgraph in the given graph such that the subgraph consists of $n$ agents. Hence, when an agent becomes an $L_{n-1}$-agent, the interacting agents make their $star$ transition to $yes$ (lines 12--14). 
When an agent $a$ with $star_a=yes$ and an agent $b$ with $star_b=no$ interact, $star_b$ transitions to $yes$ (lines 20--22). By this behavior, eventually $yes$ is conveyed to all agents. 

When agents elect multiple central agents (resp., no central agent), $star$ of each agent cannot transition to $yes$ because the central agents cannot mark $n-1$ agents (resp., there is no agent that marks agents). Thus, since initially $star$ of each agent is $no$, agents make their $star$ keep $no$. 
In both cases, since clearly the given graph is not a star, agents make a correct decision. 

\paragraph*{Decision of $never$ (lines 2--3 and 15--19 of the pseudocode)}
When two $\F'$-agents interact, they make their $star$ transition to $never$ (lines 15--16). 
When an $\F'$-agent interacts with an $l'$-agent, they also make their $star$ transition to $never$ (lines 17--19). 

From now, we show that 1) agents do not perform the above behaviors if the given graph is a star, and 2) agents perform the above behaviors or agents make their $star$ keep $no$ if the given graph is not a star. 

If the given graph is a star, agents correctly elect the central agent. Thus, since agents other than the central agent can interact only with the $L_i$-agent, agents do not perform the above behaviors.

Next, let us consider the case where the given graph is not a star. 
In this case, if agents do not confirm the star subgraph that consists of $n$ agents, agents do not make their $star$ transition to $yes$. Thus, agents make a correct decision because, if $star_a = no$ or $star_a=never$ holds, $\gamma(s_a) = no$ holds. 
On the other hand, if agents confirm the star subgraph that consists of $n$ agents, there is an agent that is not elected as the central agent and is adjacent to two or more agents. 
Moreover, after agents confirm the subgraph, the agent is an $\F'$- or $l'$-agent and is adjacent to another $\F'$- or $l'$-agent. 
Hence, eventually two $\F'$-agents interact, or the $\F'$-agent interacts with the $l'$-agent.
By the interaction, the interacting agents make their $star$ transition to $never$. 
When an agent $a$ with $star_a=never$ and an agent $b$ interact, $star_b$ transitions to $never$ (lines 2--3). By this behavior, eventually $never$ is conveyed to all agents. Note that, when this behavior occurs, the behavior of lines 20--22 does not occur. Thus, $never$ has priority over $yes$. 

From these facts, we can observe that agents make a correct decision in any case. 

\subsection*{Correctness}
First, we define some notations. 
\begin{definition}
The number of $\F$-agents is denoted by $\# \F$. 
The number of $\F'$-agents is denoted by $\# \F'$. 
The number of $L_i$-agents for $2 \le i \le n-1$ is denoted by $\# L_i$. Let $\# L = \# L_2 + \# L_3 + \cdots + \# L_{n-1}$. 
\end{definition}

\begin{definition}
A function $L(a, C)$ represents the number of agents counted by an agent $a$ in configuration $C$. Concretely, if $LF_a=L_i$ holds in $C$, $L(a, C)=i$ holds. If $LF_a \in \{\F$, $\F'$, $l' \}$ holds in $C$, $L(a,C)=0$ holds. 
\end{definition}

From now, we show an equation that holds in any execution of Algorithm \ref{alg:starn}. 
\begin{lemma}
\label{lem:FL}
Let us consider an execution $\Xi$ with some graph $G=(V, E)$, where $V=\{v_1$, $v_2$, $\ldots$, $v_n\}$. In any configuration $C$ of $\Xi$, $L(v_1, C)+ L(v_2, C)+ \cdots + L(v_n, C)= \# \F' + \# L$ holds. 
\end{lemma}
\begin{proof}
Let us consider an execution $\Xi$ with a graph $G=(V, E)$, where $V=\{v_1$, $v_2$, $\ldots$, $v_n\}$.
We prove the lemma by induction on the index of a configuration.
In the base case ($C_0$), clearly $L(v_1, C_0)+ L(v_2, C_0)+ \cdots + L(v_n, C_0)= \# \F' + \# L = 0$ holds.
For the induction step, we assume that the equation holds in $C_k$. 
Let us consider an interaction at $C_k \rightarrow C_{k+1}$ for three cases. 
\begin{itemize}
\item Case where the behavior of lines 7--8 or 9--11 does not occur at $C_k \rightarrow C_{k+1}$: In this case, clearly $L(v_i, C_{k+1})=L(v_i, C_k)$ holds for $1 \le i \le n$. Moreover, by the interaction, $\# \F'$ and $\# L$ do not change. Thus, $L(v_1, C_{k+1})+ L(v_2, C_{k+1})+ \cdots + L(v_n, C_{k+1})= \# \F' + \# L$ holds in this case. 
\item Case where the behavior of lines 7--8 occurs at $C_k \rightarrow C_{k+1}$: From the pseudocode, an $l'$-agent (resp., an $\F$-agent) transitions to an $L_2$-agent (resp., an $\F'$-agent) by the interaction. Hence, by the interaction, $L(v_1, C_{k+1})+ L(v_2, C_{k+1})+ \cdots + L(v_n, C_{k+1})=L(v_1, C_k)+ L(v_2, C_k)+ \cdots + L(v_n, C_k)+2$ holds. In addition, $\# \F'$ increases by one and $\# L$ increases by one. 
Thus, $L(v_1, C_{k+1})+ L(v_2, C_{k+1})+ \cdots + L(v_n, C_{k+1})= \# \F' + \# L$ holds in this case. 
\item Case where the behavior of lines 9--11 occurs at $C_k \rightarrow C_{k+1}$: From the pseudocode, an $L_i$-agent (resp., an $\F$-agent) transitions to an $L_{i+1}$-agent (resp., an $\F'$-agent) by the interaction. Hence, by the interaction, $L(v_1, C_{k+1})+ L(v_2, C_{k+1})+ \cdots + L(v_n, C_{k+1})=L(v_1, C_k)+ L(v_2, C_k)+ \cdots + L(v_n, C_k)+1$ holds, and $\# \F'$ increases by one. 
Thus, $L(v_1, C_{k+1})+ L(v_2, C_{k+1})+ \cdots + L(v_n, C_{k+1})= \# \F' + \# L$ holds in this case. 
\end{itemize}
In each case, the equation holds after the interaction. 
Therefore, the lemma holds. 
\end{proof}

Next, by using Lemma \ref{lem:FL}, we show that, if the given communication graph is a star, agents make a correct decision. 
\begin{lemma}
\label{lem:stary}
If the given communication graph is a star, $star_a$ of each agent $a$ converges to $yes$ for any weakly-fair execution $\Xi$.
\end{lemma}
\begin{proof}
First of all, by the property of the star graph, an agent $a$ with degree $n-1$ joins any interaction. 
In the first interaction, agent $a$ and another agent $b$ interact. From the pseudocode, agents $a$ and $b$ transition to $l'$-agents. 
After that, since agents other than $a$ and $b$ are $\F$-agents and agent $a$ with degree $n-1$ is the $l'$-agent, agents can transition to other states only by the behavior of lines 7--8. 
From weak fairness, eventually the behavior of lines 7--8 occurs, and agent $a$ transitions to an $L_2$-agent. 

After that, agent $a$ does not become an $\F$-, $\F'$-, or $l'$-agent from the pseudocode.
Thus, since agent $a$ always joins interactions, agent $b$ is always the $l'$-agent, and any agent other than $a$ and $b$ is $\F$- or $\F'$-agent. 
Hence, from Lemma \ref{lem:FL}, if agent $a$ is the $L_i$-agent, $\# \F'$ is $i-1$ and $\# \F$ is $n-2-(i-1)=n-i-1$. 
Then, from the pseudocode, agent $a$ and $\F$-agents can update their $LF$ only by the behavior of line 9--11. 
Thus, since agent $a$ is adjacent to each agent, the behavior of line 9--11 occurs repeatedly from weak fairness until agent $a$ transitions to an $L_{n-1}$-agent. Note that, since $\# \F$ is $n-i-1$ when $LF_a=L_i$ holds, eventually agent $a$ transitions to the $L_{n-1}$-agent. 
When agent $a$ transitions to the $L_{n-1}$-agent, $star_a$ transitions to $yes$ from the pseudocode. 
Then, agent $a$ is always the $L_{n-1}$-agent. Hence, since agents other than agent $a$ are adjacent only to agent $a$, $star$ of any agent does not transition to $never$ because the behaviors of lines 15--19 do not occur. 
Thus, from the pseudocode, $star$ of each agent transitions to $yes$ by the behavior of lines 20--22, and $star$ of each agent does not update afterwards. 
Therefore, the lemma holds. 
\end{proof}

From now, we show that, even if the given communication graph is not a star, agents make a correct decision. 
To show this, we first show the property of a configuration such that there is an $L_{n-1}$-agent in the configuration. 

\begin{lemma}
\label{lem:stacon}
If there is an $L_{n-1}$-agent in a configuration, there is an $l'$-agent and other $n-2$ agents are $\F'$-agents in the configuration. 
Furthermore, after the configuration, each agent $a$ does not update its $LF_a$. 
\end{lemma}
\begin{proof}
Let $C$ be a configuration in which there is an agent $a$ that is an $L_{n-1}$-agent. 

First of all, clearly $\# \F' + \# L > n$ does not hold in any configuration. 
Hence, from Lemma \ref{lem:FL}, any agent other than $a$ is not an $L_i$-agent for $2 \le i \le n-1$ in $C$ (otherwise $\# \F' + \# L = n-1+i > n$ holds). 
Thus, since $\# \F' + \# L =n-1$ holds in $C$, $n-2$ agents are $\F'$-agents and one agent $b$ is an $\F$- or $l'$-agent. 

Now, we show that agent $b$ is an $l'$-agent in $C$. 
From the pseudocode, if an agent is an $\F'$-agent in $C$, the agent is not an $l'$-agent before and after $C$. Since an $l'$-agent occurs only by the behavior of lines 5--6 and an $L_{n-1}$-agent exists in $C$, the behavior of lines 5--6 occurs before $C$. 
From these facts, agents $a$ and $b$ performed the behavior of lines 5--6 before $C$. 
From the pseudocode, if an agent is an $l'$-agent, the agent is not an $\F$-agent afterwards. 
Hence, agent $b$ is an $l'$-agent in $C$ and thus agent $a$ is the $L_{n-1}$-agent, agent $b$ is the $l'$-agent, and other $n-2$ agents are $\F'$-agents in $C$. 

From the pseudocode, clearly each agent $a$ does not update its $LF_a$ after $C$. 
Therefore, the lemma holds. 
\end{proof}

Now, by using Lemma \ref{lem:stacon}, we show that agents make a correct decision even if the given communication graph is not a star. 

\begin{lemma}
\label{lem:starn}
If the given communication graph is not a star, $star_a$ of each agent $a$ converges to $no$ or $never$ for any weakly-fair execution $\Xi$.
\end{lemma}
\begin{proof}
Let $\Xi$ be a weakly-fair execution with a non-star graph. 
Let us consider two cases: (1) $star_a$ of some agent $a$ transitions to $yes$ in $\Xi$, and (2) $star$ of any agent does not transition to $yes$ in $\Xi$. 

In case (1), since the behavior of lines 12--14 occurs, an agent $a$ becomes an $L_{n-1}$-agent. Moreover, from the pseudocode, agent $a$ is always the $L_{n-1}$-agent afterwards. 
From now on, let us consider configurations after agent $a$ becomes the $L_{n-1}$-agent. 
Since the given graph is not a star, there are two agents whose degree is two or more. 
Hence, there is an agent $b$ that is adjacent to an agent other than $a$. 
From weak fairness, eventually agent $b$ interacts with the agent other than $a$. 
By the interaction, from Lemma \ref{lem:stacon}, the behavior of lines 15--16 or 17--19 occurs and $star_b$ transitions to $never$. After that, by the behavior of lines 2--3, $star$ of each agent transitions to $never$. Furthermore, after agents update their $star$ to $never$, agents do not update their $star$. Thus, $star$ of each agent converges to $never$ in this case. 

In case (2), if $star$ of an agent transitions to $never$, $star$ of each agent converges to $never$ similarly to case (1); otherwise, since initially $star$ of each agent is $no$, $star$ of each agent converges to $no$.
Therefore, the lemma holds. 
\end{proof}

From Lemma \ref{lem:stary}, if the given communication graph is a star, $star_a$ of each agent $a$ converges to $yes$ in any execution of the protocol. 
From Lemma \ref{lem:starn}, if the given communication graph is not a star, $star_a$ of each agent $a$ converges to $no$ or $never$ in any execution of the protocol. 
Thus, we can obtain the following theorem. 
\begin{theorem}
Algorithm \ref{alg:starn} solves the star identification problem. 
That is, there exists a protocol with $O(n)$ states and designated initial states that solves the star identification problem under weak fairness if the number of agents $n$ is given. 
 \end{theorem}

\section{Impossibility Results}
\subsection{A Common Property of Graph Class Identification Protocols for Impossibility Results}
In this subsection, we present a common property that holds for protocols with designated initial states under weak fairness. 

With designated initial states under weak fairness, we assume that a protocol ${\cal P}$ solves some of the graph class identification problems.
From now, we show that, with ${\cal P}$, there exists a case where agents cannot distinguish between some different connected graphs. Note that ${\cal P}$ has no constraints for an initial knowledge (i.e., for some integer $x$, ${\cal P}$ is $\pron{x}$, $\proP{x}$, or a protocol with no initial knowledge). 

\begin{lemma}
\label{lem:impwea}
Let us consider a communication graph $G=(V, E)$, where $V=\{v_1$, $v_2$, $v_3$, $\ldots$, $v_n\}$. 
Let $G'=(V', E')$ be a communication graph that satisfies the following, where $V'=\{v'_1$, $v'_2$, $v'_3$, $\ldots$, $v'_{2n}\}$. 
\begin{itemize}
\item $E'=\{(v'_x,v'_y),(v'_{x+n},v'_{y+n}) \in V' \times V' \mid (v_x, v_y) \in E \} \cup \{(v'_1,v'_{z+n}), (v'_{1+n}, v'_z) \in V' \times V' \mid (v_1, v_{z}) \in E\}$ (Figure \ref{fig:impg} shows an example of the graphs).
\end{itemize}
Let $\Xi$ be a weakly-fair execution of ${\cal P}$ with $G$.  
If there exists a configuration $C$ of $\Xi$ after which $\forall v \in V:\gamma(s(v))=yn \in \{yes$, $no\}$ holds, there exists an execution $\Xi'$ of ${\cal P}$ with $G'$ such that there exists a configuration $C'$ of $\Xi'$ after which $\forall v' \in V':\gamma(s({v'}))=yn$ holds. 
\begin{figure}[t]
\begin{center}
\includegraphics[scale=0.35]{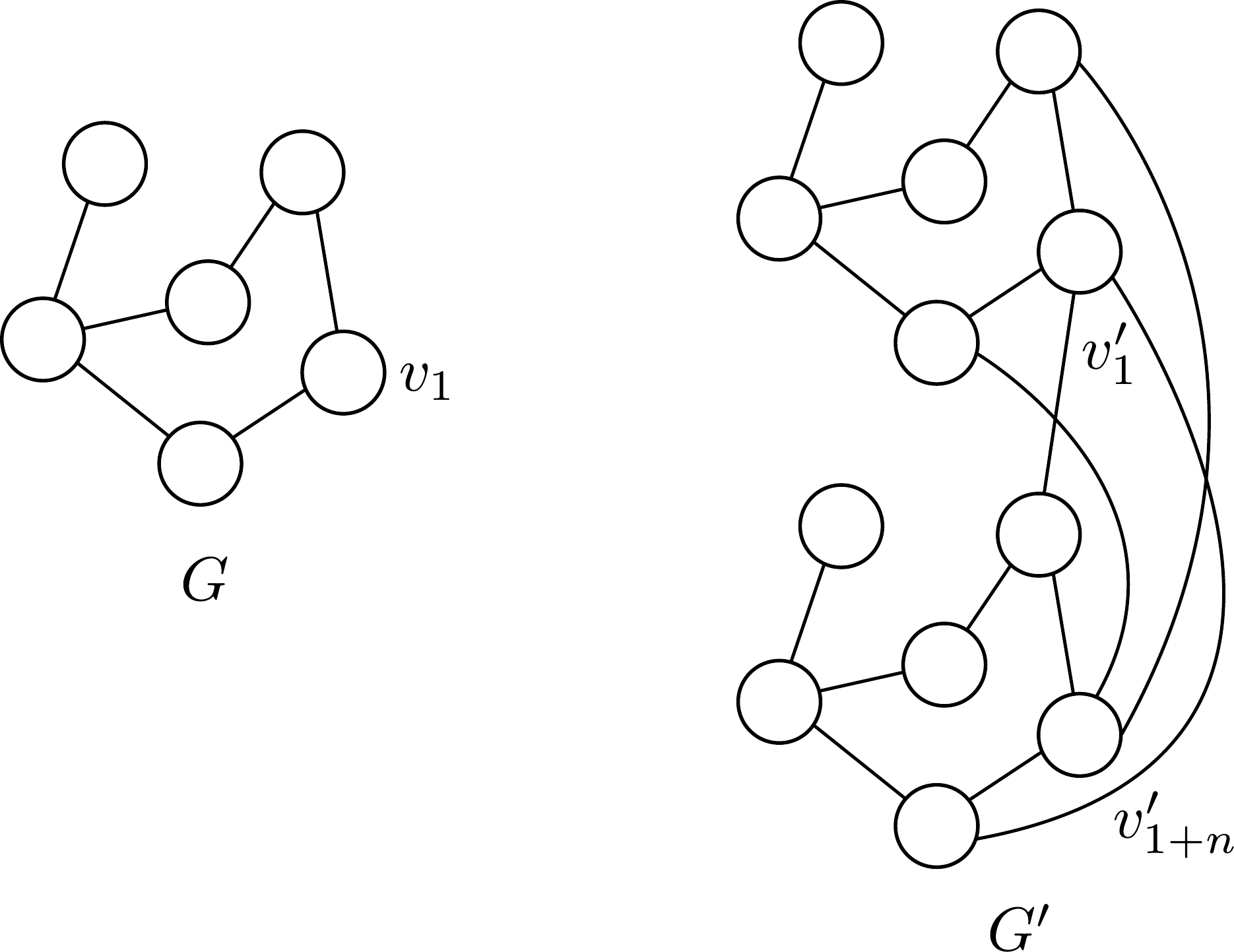}
\caption{An example of graphs $G$ and $G'$}
\label{fig:impg}
\end{center}
\end{figure}
\end{lemma}
\begin{proof}

First, we define a term that represents a relation between configurations. 
Let $C_i$ and $C'_j$ be configurations with $G$ and $G'$, respectively.
If $s(v_x, C_i)=s(v'_x, C'_j)=s(v'_{x+n}, C'_j)$ holds for $1 \le x \le n$, we say that $C_i$ and $C'_j$ are equivalent.

Let $\Xi$ be a weakly-fair execution of ${\cal P}$ with $G$ such that there exists a configuration $C_t$ of $\Xi$ after which $\forall v \in V:\gamma(s(v))=yn \in \{yes$, $no\}$ holds. Without loss of generality, we assume $yn=yes$. 

Since $\Xi$ is weakly fair, $\Xi$ can be represented by the following $\Xi = C_0$, $C_1$, $C_2$, $\ldots$, $C_t$, $\ldots$, $C_{t'_0}$, $\xi_1$, $C_{t'_1}$, $\xi_2$, $C_{t'_2}$, $\xi_3$, $\ldots$. 
\begin{itemize}
\item For $u \ge 0$, $C_{t'_u}$ is a stable configuration such that $C_{t'_0}=C_{t'_1}=C_{t'_2}=\cdots$ holds and $C_{t'_u}$ occurs infinitely often. 
\item For $j \ge 1$, $\xi_j$ is a sub-execution such that, in $C_{t'_{j-1}},\xi_{j},C_{t'_j}$, for each pair $(a, b)$ in $E$, agent $a$ interacts with agent $b$ and agent $b$ interacts with agent $a$, at least once. 
\end{itemize}

\begin{figure}[t]
\begin{center}
\includegraphics[scale=0.24]{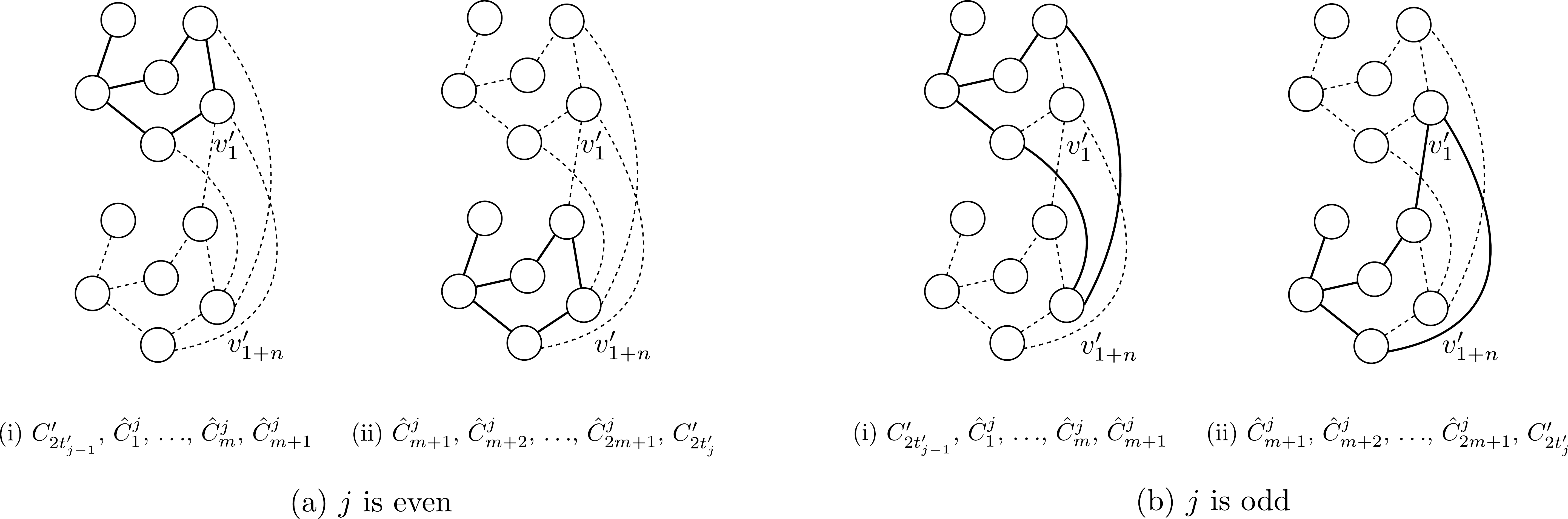}
\caption{Images of $C'_{2t'_{j-1}}$, $\xi'_j$, $C'_{2t'_j}$. A solid line represents an edge on which interactions occur in $C'_{2t'_{j-1}}$, $\xi'_j$, $C'_{2t'_j}$, and a dashed line represents an edge on which interactions do not occur in $C'_{2t'_{j-1}}$, $\xi'_j$, $C'_{2t'_j}$}
\label{fig:impe}
\end{center}
\end{figure}

Next, let us consider the following execution $\Xi' = C'_0$, $C'_1$, $C'_2$, $\ldots$, $C'_{2t'_0}$, $\xi'_1$, $C'_{2t'_1}$, $\xi'_2$, $C'_{2t'_2}$, $\xi'_3$, $\ldots$ of ${\cal P}$ with $G'$. 
\begin{itemize}
\item For $0 \le i < t'_0$, when $v_x$ interacts with $v_y$ at $C_i \rightarrow C_{i+1}$, $v'_x$ interacts with $v'_y$ at $C'_{i} \rightarrow C'_{i+1}$, and $v'_{x+n}$ interacts with $v'_{y+n}$ at $C'_{i+t'_0} \rightarrow C'_{i+t'_0+1}$. 
Clearly, $v'_1$, $\ldots$, $v'_n$ and $v'_{1+n}$, $\ldots$, $v'_{2n}$ behave similarly to $v_1$, $\ldots$, $v_n$ in $\Xi$ and thus $C_{t'_0}$ and $C'_{2t'_0}$ are equivalent. 
\item For $j \ge 1$, by using $\xi_{j} = \C{j}{1}$, $\C{j}{2}$, $\C{j}{3}$, $\ldots$, $\C{j}{m}$, we define $\xi'_{j} = \Ch{j}{1}$, $\Ch{j}{2}$, $\Ch{j}{3}$, $\ldots$, $\Ch{j}{2m+1}$ as follows (Figure \ref{fig:impe} shows images of $C'_{2t'_{j-1}}$, $\xi'_j$, $C'_{2t'_j}$). 
\begin{itemize}
\item Case where $j$ is an even number: $v'_1$, $\ldots$, $v'_n$ and $v'_{1+n}$, $\ldots$, $v'_{2n}$ behave similarly to $v_1$, $\ldots$, $v_n$ in $\Xi$. 
Concretely, for $1 \le i < m$, when $v_x$ interacts with $v_y$ at $\C{j}{i} \rightarrow \C{j}{i+1}$, $v'_x$ interacts with $v'_y$ at $\Ch{j}{i} \rightarrow \Ch{j}{i+1}$, and $v'_{x+n}$ interacts with $v'_{y+n}$ at $\Ch{j}{i+m+1} \rightarrow \Ch{j}{i+m+2}$. 
When $v_x$ interacts with $v_y$ at $C_{t'_{j-1}} \rightarrow \C{j}{1}$ (resp., $\C{j}{m} \rightarrow C_{t'_{j}}$), $v'_x$ interacts with $v'_y$ at $C'_{2t'_{j-1}} \rightarrow \Ch{j}{1}$ (resp., $\Ch{j}{m} \rightarrow \Ch{j}{m+1}$), and $v'_{x+n}$ interacts with $v'_{y+n}$ at $\Ch{j}{m+1} \rightarrow \Ch{j}{m+2}$ (resp., $\Ch{j}{2m+1} \rightarrow C'_{2t'_{j}}$). 

In this case, clearly $v'_1$, $\ldots$, $v'_n$ and $v'_{1+n}$, $\ldots$, $v'_{2n}$ make transitions similarly to $v_1$, $\ldots$, $v_n$ in $\Xi$. 
Hence, if $C_{t'_{j-1}}$ and $C'_{2t'_{j-1}}$ are equivalent, $C_{t'_{j}}$ and $C'_{2t'_{j}}$ are equivalent. 
\item Case where $j$ is an odd number: $\{v'_{1+n}\} \cup \{v'_2$, $\ldots$, $v'_n\}$ and $\{v'_1\} \cup \{v'_{2+n}$,  $\ldots$, $v'_{2n}\}$ behave similarly to $v_1$, $\ldots$, $v_n$ in $\Xi$. 
Concretely, for $1 \le i < m$ and $x, y \neq 1$, when $v_x$ interacts with $v_y$ at $\C{j}{i} \rightarrow \C{j}{i+1}$, $v'_x$ interacts with $v'_y$ at $\Ch{j}{i} \rightarrow \Ch{j}{i+1}$, and $v'_{x+n}$ interacts with $v'_{y+n}$ at $\Ch{j}{i+m+1} \rightarrow \Ch{j}{i+m+2}$. 
When $v_x$ interacts with $v_y$ at $C_{t'_{j-1}} \rightarrow \C{j}{1}$ (resp., $\C{j}{m} \rightarrow C_{t'_{j}}$), $v'_x$ interacts with $v'_y$ at $C'_{2t'_{j-1}} \rightarrow \Ch{j}{1}$ (resp., $\Ch{j}{m} \rightarrow \Ch{j}{m+1}$), and $v'_{x+n}$ interacts with $v'_{y+n}$ at $\Ch{j}{m+1} \rightarrow \Ch{j}{m+2}$ (resp., $\Ch{j}{2m+1} \rightarrow C'_{2t'_{j}}$). 

For $1 \le i < m$, when $v_1$ interacts with $v_x$ (resp., $v_x$ interacts with $v_1$) at $\C{j}{i} \rightarrow \C{j}{i+1}$, $v'_{1+n}$ interacts with $v'_x$ (resp., $v'_x$ interacts with $v'_{1+n}$) at $\Ch{j}{i} \rightarrow \Ch{j}{i+1}$, and $v'_{1}$ interacts with $v'_{x+n}$ (resp., $v'_{x+n}$ interacts with $v'_{1}$) at $\Ch{j}{i+m+1} \rightarrow \Ch{j}{i+m+2}$. 
When $v_1$ interacts with $v_x$ (resp., $v_x$ interacts with $v_1$) at $C_{t'_{j-1}} \rightarrow \C{j}{i}$, $v'_{1+n}$ interacts with $v'_x$ (resp., $v'_x$ interacts with $v'_{1+n}$) at $C'_{2t'_{j-1}} \rightarrow \Ch{j}{i}$, and $v'_{1}$ interacts with $v'_{x+n}$ (resp., $v'_{x+n}$ interacts with $v'_{1}$) at $\Ch{j}{m+1} \rightarrow \Ch{j}{m+2}$. 
When $v_1$ interacts with $v_x$ (resp., $v_x$ interacts with $v_1$) at $\C{j}{m} \rightarrow C_{t'_{j}}$, $v'_{1+n}$ interacts with $v'_x$ (resp., $v'_x$ interacts with $v'_{1+n}$) at $\Ch{j}{m} \rightarrow \Ch{j}{m+1}$, and $v'_{1}$ interacts with $v'_{x+n}$ (resp., $v'_{x+n}$ interacts with $v'_{1}$) at $\Ch{j}{2m+1} \rightarrow C'_{2t'_{j}}$. 

In this case, $\{v'_{1+n}\} \cup \{v'_2$, $\ldots$, $v'_n\}$ and $\{v'_1\} \cup \{v'_{2+n}$,  $\ldots$, $v'_{2n}\}$ make transitions similarly to $v_1$, $\ldots$, $v_n$ in $\Xi$. 
Hence, if $C_{t'_{j-1}}$ and $C'_{2t'_{j-1}}$ are equivalent, $C_{t'_{j}}$ and $C'_{2t'_{j}}$ are equivalent because $s(v_1, C_{t'_{j-1}})=s(v'_1, C'_{2t'_{j-1}})=s(v'_{1+n},C'_{2t'_{j-1}})$ holds. 
\end{itemize}
\end{itemize}

Since $\Xi$ is weakly fair, clearly each pair of agents in $E'$ interacts infinitely often in $\Xi'$ and thus $\Xi'$ satisfies weak fairness. 
By behaviors of $\Xi'$, since $C'_{2t}$ and $C_t$ are equivalent and $\forall v \in V:\gamma(s(v))=yes$ holds after $C_t$, $\forall v' \in V':\gamma(s(v'))= yes$ holds after $C'_{2t}$. 
From these facts, the lemma holds. 
\end{proof}

\subsection{Impossibility with the Known Upper Bound of the Number of Agents under Weak Fairness}

For the purpose of the contradiction, we assume that, for an integer $x$, there exists a protocol $\proP{x}$ that solves some of the graph class identification problems with designated initial states under weak fairness.
We can apply Lemma \ref{lem:impwea} to $\proP{x}$ because we can apply the same protocol $\proP{x}$ to both $G$ and $G'$ in Lemma \ref{lem:impwea}. 
Clearly, we can construct $G$ and $G'$ in Lemma \ref{lem:impwea} such that, for any of properties \emph{line}, \emph{ring}, \emph{tree}, \emph{$k$-regular}, and \emph{star}, $G$ is a graph that satisfies the property, and $G'$ is a graph that does not satisfy the property. 
Therefore, we have the following theorem. 
\begin{theorem}
\label{the:impP}
Even if the upper bound of the number of agents is given, there exists no protocol that solves the line, ring, $k$-regular, star, or tree identification problem with the designated initial states under weak fairness.
\end{theorem}

Note that, in Theorem \ref{the:impP}, the bipartite identification problem is not included. 
However, we show later that there is no protocol that solves the bipartite identification problem even if the number of agents is given. 

\subsection{Impossibility with the Known Number of Agents under Weak Fairness}
In this subsection, we show that, even if the number of agents $n$ is given, there exists no protocol that solves the \emph{line}, \emph{ring}, \emph{$k$-regular}, \emph{tree}, or \emph{bipartite} identification problem with designated initial states under weak fairness. 

\paragraph*{Case of Line, Ring, $k$-regular, and Tree}
First, we show that there exists no protocol that solves the \emph{line}, \emph{ring}, \emph{$k$-regular}, or \emph{tree} identification problem. 
Concretely, we show that there is a case where a line graph and a ring graph are not distinguishable. 
To show this, first we give some notations. 
Let $G=(V, E)$ be a line graph with four agents, where $V=\{v_1$, $v_2$, $v_3$, $v_4\}$ and $E=\{(v_1, v_2)$, $(v_2, v_3)$, $(v_3, v_4)\}$. 
Let $G'=(V', E')$ be a ring graph with four agents, where $V'=\{v'_1$, $v'_2$, $v'_3$, $v'_4 \}$ and $E'=\{(v'_1, v'_2)$, $(v'_2, v'_3)$, $(v'_3, v'_4)$, $(v'_4, v'_1)\}$.
Let $s_0$ be an initial state of agents. 
Let us consider a transition sequence $T = (s_0$, $s_0) \rightarrow (s_{a_1}$, $s_{b_1})$, $(s_{b_1}$, $s_{a_1}) \rightarrow (s_{b_2}$, $s_{a_2})$, $(s_{a_2}$, $s_{b_2}) \rightarrow (s_{a_3}$, $s_{b_3})$, $(s_{b_3}$, $s_{a_3}) \rightarrow (s_{b_4}$, $s_{a_4})$, $\ldots$.
Since the number of states is finite, there are $i$ and $j$ such that $s_{a_i}=s_{a_j}$, $s_{b_i}=s_{b_j}$, and $i<j$ hold. 
Let $sa$ and $sb$ be states such that $sa=s_{a_i}=s_{a_j}$ and $sb=s_{b_i}=s_{b_j}$ hold.

From now, we show the intuition of the proof as follows: 
We construct an execution $\Xi$ with $G$. In $\Xi$, agents repeat the following three sub-execution: 1) agents $v_1$ and $v_2$ interact repeatedly until $v_1$ and $v_2$ obtain $sa$ and $sb$, respectively, 2) agents $v_3$ and $v_4$ interact repeatedly until $v_3$ and $v_4$ obtain $sa$ and $sb$, respectively, and 3) $v_3$ and $v_2$ interact repeatedly until $v_3$ and $v_2$ obtain $sa$ and $sb$, respectively. 
Next, we construct an execution $\Xi'$ with $G'$. In $\Xi'$, agents repeat the following four sub-execution: 1) agents $v'_1$ and $v'_2$ interact repeatedly until $v'_1$ and $v'_2$ obtain $sa$ and $sb$, respectively, 2) agents $v'_3$ and $v'_4$ interact repeatedly until $v'_3$ and $v'_4$ obtain $sa$ and $sb$, respectively, 3) $v'_3$ and $v'_2$ interact repeatedly until $v'_3$ and $v'_2$ obtain $sa$ and $sb$, respectively, and 4) $v'_1$ and $v'_4$ interact repeatedly until $v'_1$ and $v'_4$ obtain $sa$ and $sb$, respectively. 
From the definition of $sa$ and $sb$, we can construct those executions, and clearly those executions satisfy weak fairness and agents make the same decision. 

Now, we show the details. First, we define a particular execution $\Xi$ with the line graph $G$. 
Concretely, we define an execution $\Xi = C_0$, $\xi_1$, $C_{u_1}$, $\xi_2$, $C_{u_2}$, $\xi_3$, $C_{u_3}$, $\xi_4$, $C_{u_4}$, $\xi_5$, $\ldots$ of a protocol ${\cal P}$ with $G$ as follows, where $\xi_m$ is a sub-execution ($m \ge 1$).
\begin{itemize}
\item In $C_{0}$, $\xi_1$, $C_{u_1}$, until $s(v_1)=sa$ and $s(v_2)=sb$ hold, agents repeat the following interactions: $v_1$ interacts with $v_2$, $v_2$ interacts with $v_1$, $v_1$ interacts with $v_2$ $\ldots$.  
From the definition of the transition sequence $T$, this is possible. 
\item In $C_{u_1}$, $\xi_2$, $C_{u_2}$, until $s(v_3)=sa$ and $s(v_4)=sb$ hold, agents repeat the following interactions: $v_3$ interacts with $v_4$, $v_4$ interacts with $v_3$, $v_3$ interacts with $v_4$ $\ldots$.
Thus, $s(v_1, C_{u_2})=s(v_3, C_{u_2})= sa$ and $s(v_2, C_{u_2})=s(v_4, C_{u_2})= sb$ hold. 
\item For $i \ge 2$, we construct the execution as follows:
\begin{itemize}
\item Case where $i \bmod 3 = 0$ holds: 
In $C_{u_i}$, $\xi_{i+1}$, $C_{u_{i+1}}$, until $s(v_1)=sa$ and $s(v_2)=sb$ hold, agents repeat the following interactions: $v_1$ interacts with $v_2$, $v_2$ interacts with $v_1$, $v_1$ interacts with $v_2$ $\ldots$.
To satisfy weak fairness, we construct the interactions so that $v_1$ and $v_2$ interact at least twice. 
\item Case where $i \bmod 3 = 1$ holds: 
In $C_{u_i}$, $\xi_{i+1}$, $C_{u_{i+1}}$, until $s(v_2)=sb$ and $s(v_3)=sa$ hold, agents repeat the following interactions: $v_3$ interacts with $v_2$, $v_2$ interacts with $v_3$, $v_3$ interacts with $v_2$ $\ldots$.
To satisfy weak fairness, we construct the interactions so that $v_2$ and $v_3$ interact at least twice. 
\item Case where $i \bmod 3 = 2$ holds: 
In $C_{u_i}$, $\xi_{i+1}$, $C_{u_{i+1}}$, until $s(v_3)=sa$ and $s(v_4)=sb$ hold, agents repeat the following interactions: $v_3$ interacts with $v_4$, $v_4$ interacts with $v_3$, $v_3$ interacts with $v_4$ $\ldots$.
To satisfy weak fairness, we construct the interactions so that $v_3$ and $v_4$ interact at least twice. 
\end{itemize}
For $i \ge 2$, if $s(v_1, C_{u_i})=s(v_3, C_{u_i})= sa$ and $s(v_2, C_{u_i})=s(v_4, C_{u_i})= sb$ hold, we can construct such interactions from the definition of the transition sequence $T$. 
Thus, since $s(v_1, C_{u_2})=s(v_3, C_{u_2})= sa$ and $s(v_2, C_{u_2})=s(v_4, C_{u_2})= sb$ hold, $C_{u_i}=C_{u_{i+1}}$ holds for $i \ge 2$. 
\end{itemize}
Since each pair of agents interact infinitely often in $\Xi$, $\Xi$ is weakly-fair.
Since $\Xi$ is weakly-fair, $\gamma(sa)=\gamma(sb)=yn \in \{yes, no\}$ holds in a stable configuration of $\Xi$. 

Now, we show that there is a case where a line graph and a ring graph are not distinguishable. 
\begin{lemma}
\label{lem:impweaN1}
There exists a weakly-fair execution $\Xi'$ of ${\cal P}$ with $G'$ such that $\forall v' \in V':\gamma(s({v'}))= yn$ holds in a stable configuration of $\Xi'$.
\end{lemma}
\begin{proof}
Let us consider the following execution $\Xi' = C'_0$, $\xi'_1$, $C'_{u'_1}$, $\xi'_2$, $C'_{u'_2}$, $\xi'_3$, $C'_{u'_3}$, $\xi'_4$, $C'_{u'_4}$, $\xi'_5$, $\ldots$ of ${\cal P}$ with $G'$, where $\xi'_m$ is a sub-execution ($m \ge 1$).
\begin{itemize}
\item In $C'_0$, $\xi'_1$, $C'_{u'_1}$, until $s(v'_1)=sa$ and $s(v'_2)=sb$ hold, agents repeat the following interactions: $v'_1$ interacts with $v'_2$, $v'_2$ interacts with $v'_1$, $v'_1$ interacts with $v'_2$ $\ldots$.
From the definition of the transition sequence $T$, this is possible. 
\item In $C'_{u'_1}$, $\xi'_2$, $C'_{u'_2}$, until $s(v'_3)=sa$ and $s(v'_4)=sb$ hold, agents repeat the following interactions: $v'_3$ interacts with $v'_4$, $v'_4$ interacts with $v'_3$, $v'_3$ interacts with $v'_4$ $\ldots$.
Thus, $s(v'_1, C'_{u'_2})=s(v'_3, C'_{u'_2})= sa$ and $s(v'_2, C'_{u'_2})=s(v'_4, C'_{u'_2})= sb$ hold. 
\item For $i \ge 2$, we construct the execution as follows: 
\begin{itemize}
\item Case where $i \bmod 4 = 0$ holds: 
In $C'_{u'_i}$, $\xi'_{i+1}$, $C'_{u'_{i+1}}$, until $s(v'_1)=sa$ and $s(v'_2)=sb$ hold, agents repeat the following interactions: $v'_1$ interacts with $v'_2$, $v'_2$ interacts with $v'_1$, $v'_1$ interacts with $v'_2$ $\ldots$.
To satisfy weak fairness, we construct the interactions so that $v'_1$ and $v'_2$ interact at least twice. 
\item Case where $i \bmod 4 = 1$ holds: 
In $C'_{u'_i}$, $\xi'_{i+1}$, $C'_{u'_{i+1}}$, until $s(v'_2)=sb$ and $s(v'_3)=sa$ hold, agents repeat the following interactions: $v'_3$ interacts with $v'_2$, $v'_2$ interacts with $v'_3$, $v'_3$ interacts with $v'_2$ $\ldots$.
To satisfy weak fairness, we construct the interactions so that $v'_2$ and $v'_3$ interact at least twice. 
\item Case where $i \bmod 4 = 2$ holds: 
In $C'_{u'_i}$, $\xi'_{i+1}$, $C'_{u'_{i+1}}$, until $s(v'_3)=sa$ and $s(v'_4)=sb$ hold, agents repeat the following interactions: $v'_3$ interacts with $v'_4$, $v'_4$ interacts with $v'_3$, $v'_3$ interacts with $v'_4$ $\ldots$.
To satisfy weak fairness, we construct the interactions so that $v'_3$ and $v'_4$ interact at least twice. 
\item Case where $i \bmod 4 = 3$ holds: 
In $C'_{u'_i}$, $\xi'_{i+1}$, $C'_{u'_{i+1}}$, until $s(v'_4)=sb$ and $s(v'_1)=sa$ hold, agents repeat the following interactions: $v'_1$ interacts with $v'_4$, $v'_4$ interacts with $v'_1$, $v'_1$ interacts with $v'_4$ $\ldots$.
To satisfy weak fairness, we construct the interactions so that $v'_1$ and $v'_4$ interact at least twice. 
\end{itemize}
For $i \ge 2$, if $s(v'_1, C'_{u'_i})=s(v'_3, C'_{u'_i})= sa$ and $s(v'_2, C'_{u'_i})=s(v'_4, C'_{u'_i})= sb$ hold, we can construct such interactions from the definition of the transition sequence $T$. 
Thus, since $s(v'_1, C'_{u'_2})=s(v'_3, C'_{u'_2})= sa$ and $s(v'_2, C'_{u'_2})=s(v'_4, C'_{u'_2})= sb$ hold, $C'_{u'_i}=C'_{u'_{i+1}}$ holds for $i \ge 2$. 
\end{itemize}

Since each pair of agents interact infinitely often in $\Xi'$, $\Xi'$ is weakly-fair.
From these facts, since $\gamma(sa)=\gamma(sb)=yn$ holds, $\forall v' \in V':\gamma(s({v'}))=yn$ holds in a stable configuration in $\Xi'$. 
Thus, the lemma holds.
\end{proof}

Note that, even if the number of agents is given, Lemma \ref{lem:impweaN1} holds because $|V|=|V'|=4$ holds in the lemma. 
In Lemma \ref{lem:impweaN1}, $G$ is a line graph and a tree graph whereas $G'$ is neither a line graph nor a tree graph. 
Furthermore, $G'$ is a ring graph and a $2$-regular graph whereas $G$ is neither a ring graph nor a $2$-regular graph. 
Hence, by Lemma \ref{lem:impweaN1}, there is no protocol that solves the \emph{line}, \emph{ring}, \emph{tree}, or \emph{$k$-regular} identification problem, and thus we have the following theorem. 
\begin{theorem}
Even if the number of agents $n$ is given, there exists no protocol that solves the line, ring, $k$-regular, or tree identification problem with designated initial states under weak fairness.
\end{theorem}

\paragraph*{Case of Bipartite}
Next, we show that there exists no protocol that solves the bipartite identification problem. 
For the purpose of the contradiction, we assume that there exists a protocol $\pron{6}$ that solves the bipartite identification problem with designated initial states under weak fairness if the number of agents $6$ is given.

We define a ring graph $G=(V, E)$ with three agents, a ring graph $G'=(V', E')$ with 6 agents, and a graph $G''=(V'', E'')$ with 6 agents as follows: 
\begin{itemize}
\item $V=\{v_1$, $v_2$, $v_3\}$ and $E=\{(v_1, v_2)$, $(v_2, v_3)$, $(v_3, v_1)\}$.
\item $V'=\{v'_1$, $v'_2$, $v'_3$, $v'_4$, $v'_5$, $v'_6\}$ and $E'=\{(v'_1, v'_2)$, $(v'_2, v'_6)$, $(v'_6, v'_4)$, $(v'_4, v'_5)$, $(v'_5, v'_3)$, $(v'_3, v'_1)\}$.
\item $V''=\{v''_1$, $v''_2$, $v''_3$, $v''_4$, $v''_5$, $v''_6\}$ and $E''=\{(v''_x,v''_y),(v''_{x+n},v''_{y+n}) \in V'' \times V'' \mid (v_x, v_y) \in E \} \cup \{(v''_1,v''_5)$, $(v''_1,v''_6)$, $(v''_4,v''_2)$, $(v''_4,v''_3)\}$.
\end{itemize}
Figure \ref{fig:impgbi} shows graphs $G$, $G'$, and $G''$.


From now, we show that there exists an execution $\Xi''$ of $\pron{6}$ with $G''$ such that all agents converge to $yes$ whereas $G''$ does not satisfy \emph{bipartite}.
To show this, we first show that, in any execution $\Xi$ of $\pron{6}$ with $G$ (i.e., the protocol for 6 agents is applied to a population consisting of 3 agents), all agents converge to $yes$. 
To prove this, we borrow the proof technique in~\cite{fischer2006self}. In~\cite{fischer2006self}, Fischer and Jiang proved the impossibility of leader election for a ring graph. 
\begin{lemma}
\label{lem:impweaN2}
In any weakly-fair execution $\Xi$ of $\pron{6}$ with $G$, all agents converge to $yes$. 
That is, in $\Xi$, there exists $C_{t}$ such that $\forall v \in V:\gamma(s(v,C_i))=yes$ holds for $i \ge t$. 
\begin{figure}[t]
\begin{center}
\includegraphics[scale=0.4]{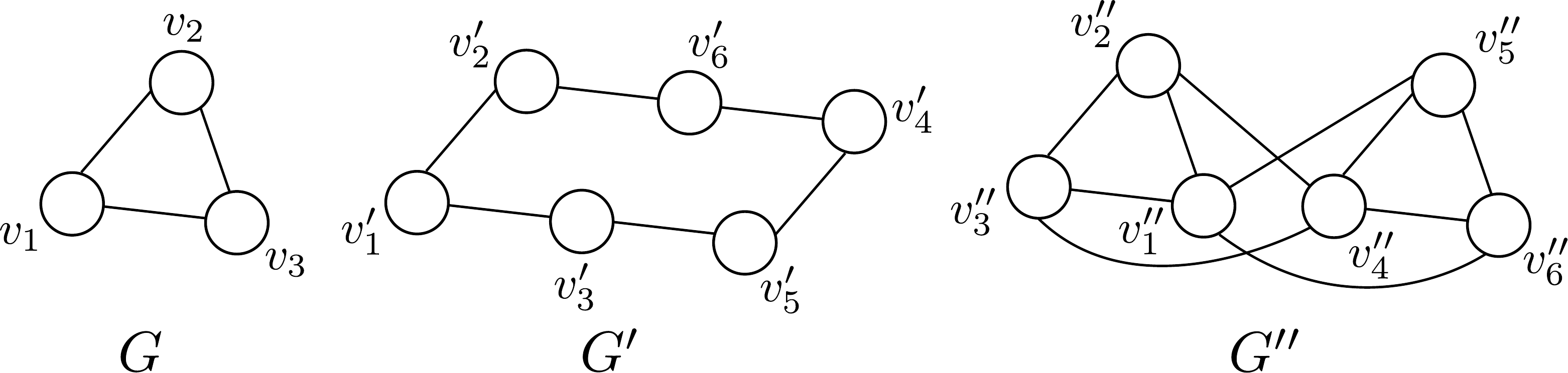}
\caption{Graphs $G$, $G'$, and $G''$}
\label{fig:impgbi}
\end{center}
\end{figure}
\end{lemma}
\begin{proof}

Let $\Xi = C_0$, $C_1$, $C_2$ be a weakly-fair execution of $\pron{6}$ with $G$. 
Let us consider the following execution $\Xi' =D_0$, $D'_0$, $D_1$, $D'_1$ $\ldots$ of $\pron{6}$ with $G'$. 

\begin{itemize}
\item For $x$ and $y$ such that either $x = 1$ or $y = 1$ holds, when $v_x$ interacts with $v_y$ at $C_i \rightarrow C_{i+1}$, $v'_x$ interacts with $v'_y$ at $D_i \rightarrow D'_{i}$, and $v'_{x+3}$ interacts with $v'_{y+3}$ at $D'_i \rightarrow D_{i+1}$.
\item For $x \in \{2, 3\}$ and $y \in \{2, 3\}$ such that $x \neq y$ holds, when $v_x$ interacts with $v_y$ at $C_i \rightarrow C_{i+1}$, $v'_x$ interacts with $v'_{y+3}$ at $D_i \rightarrow D'_{i}$, and $v'_{x+3}$ interacts with $v'_y$ at $D'_i \rightarrow D_{i+1}$.

\end{itemize}

For a configuration $C$ of $G$ and a configuration $D$ of $G'$, if $s(v_i, C) = s(v'_i, D) = s(v'_{i+3}, D)$ holds for $i$ $(1 \le i \le 3)$, we say that $C$ and $D$ are $equivalent$.

From now, we show, by induction on the index of configuration, that $C_r$ and $D_r$ are equivalent for any $r \ge 0$.
Since clearly $C_0$ and $D_0$ are equivalent, the base case holds.
For the induction step, we assume that $C_l$ and $D_l$ are equivalent, and then consider two cases of interaction at $C_l \rightarrow C_{l+1}$.

First we consider the case where, for $x$ and $y$ such that either $x = 1$ or $y = 1$ holds, agents $v_x$ interacts with $v_y$ at $C_l \rightarrow C_{l+1}$.
In this case, $v'_x$ interacts with $v'_y$ at $D_l \rightarrow D'_{l}$, and $v'_{x+3}$ interacts with $v'_{y+3}$ at $D'_l \rightarrow D_{l+1}$.
By the induction assumption, $s(v_x, C_l) = s(v'_x, D_l) = s(v'_{x+3}, D_l)$ and $s(v_y, C_l) = s(v'_y, D_l) = s(v'_{y+3}, D_l)$ hold.
Hence, agents $v'_{x}$ and $v'_{x+3}$ (resp., $v'_y$ and $v'_{y+3}$) update their states similarly to $v_x$ (resp., $v_y$), and thus $C_{l+1}$ and $D_{l+1}$ are equivalent in this case.

Next, we consider the case where, for $x \in \{2, 3\}$ and $y \in \{2, 3\}$ such that $x \neq y$ holds, agents $v_x$ interacts with $v_y$ at $C_l \rightarrow C_{l+1}$.
In this case, $v'_x$ interacts with $v'_{y+3}$ at $D_l \rightarrow D'_{l}$, and $v'_{x+3}$ interacts with $v'_y$ at $D'_l \rightarrow D_{l+1}$. 
By the induction assumption, $s(v_x, C_l) = s(v'_x, D_l) = s(v'_{x+3}, D_l)$ and $s(v_y, C_l) = s(v'_y, D_l) = s(v'_{y+3}, D_l)$ hold.
Hence, agents $v'_x$ and $v'_{x+3}$ (resp., $v'_y$ and $v'_{y+3}$) update their states similarly to $v_x$ (resp., $v_y$), and thus $C_{l+1}$ and $D_{l+1}$ are equivalent in this case.
Thus, $C_r$ and $D_r$ are equivalent for any $r \ge 0$.

In $\Xi'$, since the number of agents is given correctly, a stable configuration exists, and $\forall v' \in V':\gamma(s({v'}))=yes$ holds in the configuration because $G'$ satisfies \emph{bipartite}. 
Since $C_r$ and $D_r$ are equivalent for any $r \ge 0$, $\forall v \in V:\gamma(s(v))=yes$ holds after some configuration in $\Xi$. 
Thus, the lemma holds.
\end{proof}

Now, we show that there exists execution $\Xi''$ of $\pron{6}$ with $G''$ such that all agents converge to $yes$.

\begin{lemma}
\label{lem:impweaN2-2}
With the designated initial states, there exists a weakly-fair execution $\Xi''$ of $\pron{6}$ with $G''$ such that $\forall v'' \in V'':\gamma(s({v''}))=yes$ in a stable configuration. 
\end{lemma}
\begin{proof}
By Lemma \ref{lem:impweaN2}, there exists a weakly-fair execution of $\pron{6}$ with $G$ such that all agents converge to $yes$ in the execution even if $2n$ is given as the number of agents whereas the number of agents is $n$. 
This implies that we can apply Lemma \ref{lem:impwea} to protocol $\pron{6}$ and graphs $G$ and $G''$. This is because $G$ and $G''$ satisfy the condition of $G$ and $G'$ in Lemma \ref{lem:impwea}, and the protocol $\pron{6}$ satisfies the condition of protocol ${\cal P}$ in Lemma \ref{lem:impwea}. 
Hence, there exists a weakly-fair execution $\Xi''$ of $\pron{6}$ with $G''$ such that $\forall v'' \in V':f(s(v''))=yes$ holds in a stable configuration of $\Xi''$. 
Thus, the lemma holds. 
\end{proof}

Graph $G''$ does not satisfy \emph{bipartite}. 
Thus, from Lemma \ref{lem:impweaN2-2}, $\pron{6}$ is incorrect. 
Therefore, we have the following theorem. 
\begin{theorem}
Even if the number of agents $n$ is given, there exists no protocol that solves the bipartite identification problem with the designated initial states under weak fairness.
\end{theorem}

\subsection{Impossibility with Arbitrary Initial States}
In this subsection, we show that, even if the number of agents $n$ is given, there exists no protocol that solves the \emph{line}, \emph{ring}, \emph{$k$-regular}, \emph{star}, \emph{tree}, or \emph{bipartite} identification problem with arbitrary initial states under global fairness. 

For the purpose of the contradiction, we assume that there exists a protocol ${\cal P}$ that solves some of the above graph class identification problems with arbitrary initial states under global fairness if the number of agents $n$ is given.
From now, we show that there are two executions $\Xi$ and $\Xi'$ of ${\cal P}$ such that the decision of all agents in the executions converges to the same value whereas $\Xi$ and $\Xi'$ are for graphs $G$ and $G' (\neq G)$, respectively.
\begin{lemma}
\label{lem:imparb}
Let $G=(V, E)$ and $G'=(V', E')$ be connected graphs that satisfy the following condition, where $V=\{v_1$, $v_2$, $v_3$, $\ldots$, $v_n\}$ and $V'=\{v'_1$, $v'_2$, $v'_3$, $\ldots$, $v'_n\}$.
\begin{itemize}
\item For some edge $(v_{\alpha},v_{\beta})$ in $E$, $E'= \{(v'_x,v'_y) \in V' \times V' \mid (v_x, v_y) \in E \} \backslash \{(v'_{\alpha},v'_{\beta})\}$.
\end{itemize}
If there exists a globally-fair execution $\Xi$ of ${\cal P}$ with $G$ such that $\forall v \in V:\gamma(s(v))= yn \in \{yes$, $no\}$ holds in a stable configuration of $\Xi$, there exists a globally-fair execution $\Xi'$ of ${\cal P}$ with $G'$ such that $\forall v' \in V':\gamma(s({v'}))= yn$ holds in a stable configuration of $\Xi'$.
\end{lemma}
\begin{proof}
Let $\Xi = C_0$, $C_1$, $C_2$, $\ldots$ be a globally-fair execution of ${\cal P}$ with $G$ such that $\forall v \in V:\gamma(s(v))= yn \in \{yes$, $no\}$ holds in a stable configuration. Let $C_t$ be a stable configuration in $\Xi$.
For the purpose of the contradiction, we assume that there exists no execution of ${\cal P}$ with $G'$ such that $\forall v' \in V':\gamma(s({v'}))=yn$ holds in a stable configuration. 

Let us consider an execution $\Xi'=C'_0$, $C'_1$, $C'_2$, $\ldots$, $C'_{t'}$, $\ldots$ of ${\cal P}$ with $G'$ as follows:
\begin{itemize}
\item For $1 \le i \le n$, $s(v'_i, C'_0)=s(v_i, C_t)$ holds.
\item $C'_{t'}$ is a stable configuration. 
\end{itemize}
By the assumption, $\exists v'_z \in V':\gamma(s(v'_z, C'_{t'}))= yn' (\neq yn)$ holds. 

Next, let us consider an execution $\Xi''=C''_0$, $C''_1$, $C''_2$, $\ldots$, $C''_{t}$, $\ldots$ of ${\cal P}$ with $G$ as follows:
\begin{itemize}
\item For $0 \le i \le t$, $C''_i=C_i$ holds (i.e., agents behave similarly to $\Xi$).
\item For $t < i \le t+t'$, when $v'_x$ interacts with $v'_y$ at $C'_{i-t-1} \rightarrow C'_{i-t}$, $v_x$ interacts with $v_y$ at $C''_{i-1} \rightarrow C''_{i}$.
This is possible because $E' \subset E$ holds. 
\end{itemize}
Since $C_t$ is a stable configuration, $C''_{t}$ is also a stable configuration and $\forall v \in V:\gamma(s(v, C''_t))= yn$ holds.
Since agents behave similarly to $\Xi'$ after $C''_{t}$,  $\gamma(s(v_z, C''_{t+t'}))= yn'$ holds. 
This contradicts the fact that $C''_{t}$ is a stable configuration.
\end{proof}

We can construct a non-line graph, a non-ring graph, a non-star graph, and a non-tree graph by adding an edge to a line graph, a ring graph, a star graph, and a tree graph, respectively. 
Moreover, we can construct a bipartite graph, and a $k$-regular graph by adding an edge to some non-bipartite graph, and some non-$k$-regular graph, respectively. 
From Lemma \ref{lem:imparb}, there is a case where the decision of all agents converges to the same value for each pair of graphs. 
Therefore, we have the following theorem. 
\begin{theorem}
There exists no protocol that solves the line, ring, $k$-regular, star, tree, or bipartite identification problem with arbitrary initial states under global fairness. 
\end{theorem}

\section{Concluding Remarks}
In this paper, we consider the graph class identification problems on various assumptions.
With designated initial states, we propose graph class identification protocols for trees, $k$-regular graphs, and stars. In particular, the star identification protocol works even under weak fairness. 
On the other hand, for lines, rings, $k$-regular graphs, trees, or bipartite graphs, we prove that there is no protocol under weak fairness. 
With arbitrary initial states, we prove that there is no graph class identification protocol for lines, rings, $k$-regular graphs, stars, trees, or bipartite graphs. 
We have interesting open problems for future researches as follows:
\begin{itemize}
 \item What is the time complexity of graph class identification problems?
 \item Are there some graph class identification protocols for other graph properties such as lattice, torus, and cactus?
\end{itemize}

\bibliographystyle{plain}
\bibliography{ref}

\end{document}